\documentclass{article} 

\usepackage{graphicx}

\usepackage{lineno,hyperref}
\usepackage{graphics}
\usepackage{amsmath}
\usepackage{amsthm}
\usepackage{amssymb}
\usepackage{mathtools}
\usepackage{bm}
\usepackage{authblk}
\usepackage{color}

 \newtheorem{theorem}{Theorem}[section]
  \newtheorem{lemma}[theorem]{Lemma}
  \newtheorem{condition}{Condition}
  \newtheorem{remark}[theorem]{Remark}

\begin{document}

\title{A Small Delay and Correlation Time Limit of Stochastic Differential Delay Equations with                
            State-Dependent Colored Noise}

\date{}

\author[1]{Scott Hottovy}         
\author[2,3]{Austin McDaniel}      
\author[2]{Jan Wehr}

\affil[1]{Department of Mathematics, United States Naval Academy, Annapolis, MD, 21401, USA}
\affil[2]{Department of Mathematics, University of Arizona, Tucson, AZ 85721, USA}
\affil[3]{School of Mathematical and Statistical Sciences, Arizona State University, 
                    Tempe, AZ 85287, USA}

\maketitle

\begin{abstract}
We consider a general stochastic differential delay equation (SDDE) with state-dependent colored noises and derive its limit as the time delays and the correlation times of the noises go to zero.  The work is motivated by an experiment involving an electrical circuit with noisy, delayed feedback.  An Ornstein-Uhlenbeck process is used to model the colored noise. The main methods used in the proof are a theorem about convergence of solutions of stochastic differential 
equations by Kurtz and Protter and a maximal inequality for sums of a stationary sequence of random variables by Peligrad and Utev. 
\end{abstract}

\section{Introduction}
\label{sec:Introduction}

Stochastic differential equations (SDEs) are frequently used to describe the dynamics of physical and biological systems \cite{oksendal}.  However, in situations where a system's response to stimuli is delayed, stochastic differential delay equations (SDDEs) provide more accurate models.  This gain in accuracy comes at a price of greater mathematical difficulty because the theory of SDDEs is much less developed than the theory of SDEs.  Thus, it is useful to develop approximations of SDDEs that are easier to work with than the original equations but still account for the effects of the delay(s).  Such approximations have been used recently to show a phase transition in the collective behavior of robots with sensorial delays \cite{Mijalkov} and a crossover from the It\^o to the Stratonovich equation in the dynamics of an electrical circuit \cite{pesce2013} (see \cite{Volpe2016} for a review).

In this article we consider a general SDDE system, driven by colored (i.e. temporally correlated) noise, which was motivated by an experiment involving an electrical circuit with a delayed feedback mechanism \cite{pesce2013}.   In the experiment, the voltage is driven by a colored noise process with a rapidly decaying correlation function.  An SDE approximation of an SDDE that generalizes the circuit was derived by first performing a Taylor expansion to first order in the time delays and then taking the limit in which the time delays and the correlation times of the noises go to zero at the same rate \cite{pesce2013}. This limiting SDE contains \emph{noise-induced drift} terms which depend on the ratios of the time delays to the noise correlation times.  Convergence of the solution of the equation obtained by Taylor expansion to the solution of the limiting equation was proven later in \cite{mcdaniel2016}.

The present article improves upon the approximation contained in \cite{mcdaniel2016,pesce2013}.  We study the same limit of the SDDE system, but without first performing a Taylor approximation, and thereby get a more accurate result.  Because we do not perform a Taylor expansion, we are able to use a simpler (less smooth) model of the colored noise process than the one used in \cite{mcdaniel2016,pesce2013}; here, we model the colored noise as a stationary solution to an Ornstein-Uhlenbeck SDE \cite{gardiner}.  Our refinement of the results of \cite{mcdaniel2016,pesce2013} can be used in practice as an SDE approximation of a general class of systems with delay.

We consider the multidimensional SDDE system
\begin{equation} \label{SDDE}
d \bm{x} _t = \bm{f}(\bm{x} _t) dt + \bm{\sigma}(\bm{x}_{t - \bm{\delta}}) \bm{\xi} _t dt
\end{equation}
where $\bm{x} _t\in\mathbb{R}^m$ is the state vector, $\bm{f}: \mathbb{R}^m \rightarrow \mathbb{R}^m$ is a vector-valued function describing the deterministic part of the dynamical system, $\bm{\xi} _t\in\mathbb{R}^n$ is a vector of zero-mean independent noises, $\bm{\sigma} : \mathbb{R}^m \rightarrow \mathbb{R}^{m\times n}$
is a matrix-valued function, and $\bm{x} _{t - \bm{\delta}} = \big( (x _{t - \delta _1})_1, ..., (x _{t - \delta _i})_i , ..., (x _{t - \delta _m})_m \big) ^{\rm T}$ (where ${\rm T}$ denotes transpose) is the delayed state vector.
Note that each component is delayed by a possibly different amount $\delta _i \geq 0$. Each of the $n$ independent noises $\xi _j$ is colored and therefore characterized by a correlation time $\tau_j > 0$. That is, assuming stationarity, $E \big[ ({\xi}_t)_j(\xi_s)_j \big] = g \big( |t-s|\tau^{-1}_j \big) $ where $g$ is a function that decays quickly as its argument increases (for $i \neq j$, $E \big[ ({\xi}_t)_i(\xi_s)_j \big] = 0$ by independence).

In the main theorem of the article (Theorem~\ref{theorem}), we consider the case where the components of $\bm{\xi}$ are independent Ornstein-Uhlenbeck colored noises with correlation times $\tau_j$.  That is, we define $\xi _j = y_j$ where $y_j$ is the solution of 
\begin{equation}
 \label{yequation}
d(y _t)_j = - \frac{1}{\tau _j} (y _t)_j dt + \frac{1}{\tau_j} d(W _t)_j
\end{equation}
where $\tau _j > 0$ and $\bm{W}$ is an $n$-dimensional Wiener process.  Equation~\eqref{yequation} has a unique stationary measure and for an arbitrary initial condition the distribution of $\bm{y}_t$ converges to this stationary measure as $t \to \infty$.  The solution of \eqref{yequation} with the initial condition distributed according to the stationary measure defines a stationary process whose realizations will play the role of colored noise in the SDDE system (\ref{SDDE}).  Note that as $\tau _j \rightarrow 0$ for all $j=1,...,n$, $\bm{y}$ converges to an n-dimensional white noise (i.e., its correlation function converges to a delta function).

We study the limit of the system consisting of equations (\ref{SDDE}) and (\ref{yequation}), with $\bm{\xi}_t = \bm{y}_t$, 
assuming that all $\delta _i$ and $\tau _j$ stay proportional to a single characteristic time $\epsilon > 0$ which goes to $0$. 
Thus, we let $ \delta _i = c _i \epsilon \to 0$ and $\tau _j = k_j \epsilon \to 0$ where $c _i ,k _j > 0$ remain constant for all $i,j$ and $\epsilon \to 0$. 

We consider the solution to equations (\ref{SDDE}) and (\ref{yequation}) on a bounded time interval $0 \leq t \leq T$.   We let $(\Omega, \mathcal{F}, P)$ denote the underlying probability space.  We will use the notation $\delta ^* = \max _{1 \leq i \leq m} \delta _i$ and $c^* = \max _{1 \leq i \leq m} c_i = \delta ^* / \epsilon$.  In order to formulate a well-posed problem, because of the delays in (\ref{SDDE}), one needs to specify not only an initial condition but also the values of the process $\bm{x}$ at all past times $t \in [- \delta ^* , 0]$.  Therefore, we assume that there is a $t_- < 0$ such that the values of $\bm{x}$ are initially specified for $t \in [t_-, 0]$ and we only consider delays $\delta _i$ such that $\delta _i < | t_- |$ for all $i$.  Let $\bm{x}^- : \Omega \times [t_-, 0] \rightarrow \mathbb{R}^m$ denote this $\emph{past condition}$ associated with (\ref{SDDE}).  We assume that $\bm{x}^-$ is independent of $\bm{W} _t$ for all $t \geq 0$.  We further assume that $\bm{x}^-$ is defined so that there exists a unique solution to (\ref{SDDE}) with the past condition $\bm{x}^-$.\\

\begin{theorem}
\label{theorem}
Suppose that $\bm{f}$ is continuous and bounded and that $\bm{\sigma}$ is bounded with bounded, Lipschitz continuous first derivatives.  Let $(\bm{x}^{\epsilon} , \bm{y} ^{\epsilon} )$ solve equations (\ref{SDDE}) and (\ref{yequation}) (which depend on $\epsilon$ through $\delta _i, \tau _j$) on $0 \leq t \leq T$ with the past condition $\bm{x} ^-$ the same for every $\epsilon$ and the initial condition $\bm{y} ^{\epsilon} _0$ distributed according to the stationary distribution corresponding to (\ref{yequation}).  Let $\bm{x}$ solve
\begin{equation}
\label{limiting equation}
d\bm{x} _t = \bm{f }(\bm{x}_t) dt +  \bm{S}(\bm{x}_t)\;dt + \bm{\sigma }(\bm{x}_t) d\bm{W}_t 
\end{equation}
on $0 \leq t \leq T$ with the same initial condition $\bm{x}_0=\bm{x}^-_0$, where $\bm{S}$ is 
defined componentwise as
\begin{equation}
\label{eq:Sdef}
	S_i(\bm{x}) = \sum _{p,j} \frac{1}{2} e^{- \frac{\delta _p}{\tau _j}} \sigma_{pj} (\bm{x})  \frac{\partial \sigma _{ij} }{\partial x_p} (\bm{x})
\end{equation}
and suppose strong uniqueness holds on $0 \leq t \leq T$ for (\ref{limiting equation}) with the initial condition $\bm{x}_0$.  Then
\begin{equation}
\lim _{\epsilon \rightarrow 0} P \left[ \sup_{0 \leq t \leq T} \| \bm{x}^{\epsilon}_t - \bm{x} _t \| > a \right] = 0
\end{equation}
for every $a>0$.
\end{theorem}

\vspace{5pt}

An outline of the paper is as follows. In Section \ref{sec:colored noise} we discuss the important modeling 
aspects of a general colored noise process $\bm{\xi}$. In Section \ref{sec:proof of theorem} we prove the main
theorem of the paper. In Section \ref{sec:Discussion}, we discuss how the theorem relates to the previous works \cite{mcdaniel2016,pesce2013}.
We give conclusions in Section \ref{sec:Conclusion}.

\section{Colored noise process}\label{sec:colored noise}

In this section we discuss the specific model of colored noise that we use in this paper, i.e., the Ornstein-Uhlenbeck (OU)  process.  The noise process $\bm{\xi}$ driving the system (\ref{SDDE}) is colored, not white. The terms ``colored'' and ``white'' come from 
the Wiener-Khinchin theorem (see \cite[Section 1.5.2]{gardiner}).  This theorem states that the expected value of the modulus of the Fourier transform of the time series of the stationary noise process is equal to the Fourier transform of its time correlation function.  Thus, a white noise process has a 
constant frequency correlation function in the Fourier domain (much like white light contains all colors of the light spectrum in equal proportions).  A colored noise is 
any (usually mean zero) process $\xi$ with a nonconstant frequency correlation function.  In this paper, we are interested in stationary colored noise processes $\xi$ which have a time correlation function of the form 
\begin{equation*}
	E \big[ \xi_t\xi_s \big] = \frac{1}{\epsilon} g\left (\frac{ |t-s|}{\epsilon}\right ) 
\end{equation*}
where $g$ is a function that decays rapidly as its argument increases.  For small 
$\epsilon>0$, the correlation function can be approximated by the correlation function of white noise, i.e.,
\begin{equation*}
	E \big[ \xi_t\xi_s \big] \rightarrow \delta(t-s)  
\end{equation*}
as $\epsilon\rightarrow0$.  Furthermore, it is typical to use as a model for colored noise a process $\xi$ such that 
\begin{equation*}
	\int_0^t\xi_s\;ds \rightarrow W_t 
\end{equation*}
as $\epsilon\rightarrow0$, in some probabilistic sense.
 
Along with the correlation function $g(r)$, $r\geq0$, the smoothness of 
the noise process is another important consideration when modeling physical noise.  For a mean-zero Gaussian process, the 
covariance function determines the entire law of the process, and hence also the smoothness of its realizations.  Different processes with varying degrees of smoothness have been used for modeling colored noise in the 
literature.  An infinitely differentiable process was used in \cite{freidlin2004} by convoluting a smooth function $h\in C^{\infty}([0,\infty))$
with a Wiener process.  A piecewise differentiable approximation of a Wiener process was constructed in \cite{kurtz91}
which was then differentiated to obtain a colored noise.  In \cite{mcdaniel2016,pesce2013}, a differentiable harmonic noise process 
was used in order to make the solution twice differentiable and thus allow to use its Taylor expansion.

In this paper, we use an Ornstein-Uhlenbeck (OU) process to model the colored noise.  This is an often used simple choice (see, e.g.,\;\cite{hottovyEPL2012,kupferman2004,pavliotis2005,pavliotis}) because it is a continuous process that can be written as the solution to a linear SDE.  An OU process $y$ with paths in  $C([0,T] , \mathbb{R})$ is defined as the solution to the SDE
\begin{equation}	
	\label{eq:generalOU}
	dy_t = -\alpha y_t\;dt + \sigma dW_t \; , \quad y_0 = w \; ,
\end{equation}
where $\alpha,\sigma>0$.  This process $y$ is a Gaussian process with mean $E[y_t] = 
e^{-\alpha t}E[w]$ and autocovariance function \cite[Section 4.5.4]{gardiner}
\begin{equation}
	E \bigg[ \Big( y_t-E[y_t] \Big) \Big( y_s-E[y_s] \Big) \bigg] 
	= \Bigg[ \mathrm{Var}(w) - \frac{\sigma ^2}{2 \alpha} \Bigg] e^{- \alpha (t + s)} + \frac{\sigma^2}{2\alpha}e^{-\alpha |t-s|} \; .
\end{equation}
Furthermore, the OU process is ergodic, that is, it has a unique stationary measure, such that starting from any initial distribution, the system's distribution converges to this measure.  In Section \ref{sec:Introduction}, we defined the colored noises to be $n$ independent, stationary OU processes with correlation times $\tau_j= k_j \epsilon$. In 
other words, $(\xi_t)_j = (y_t^\epsilon)_j$ is the solution to (\ref{eq:generalOU}) with 
$\alpha = \sigma = \tau _j ^{-1}$ and with $(y_0^\epsilon)_j = w_j$ where $w_j$ is a random variable, independent of the process $W$, having this stationary 
distribution, i.e., $w_j$ is Gaussian with mean zero and variance $1/(2k_j \epsilon)$.  

The choice of this particular colored noise is advantageous because equation (\ref{eq:generalOU}) is exactly solvable and every moment of $\bm{y}_t^\epsilon$ can be calculated. Furthermore, the covariance moments $E \left[ \big( (y_t^\epsilon)_j (y_s^{\epsilon})_{j} \big)^n \right]$ can be calculated for all $n\in\mathbb{N}$ by using Wick's theorem \cite[Theorem 1.3.8]{janson}. 
%
%
\begin{figure}[h!]
\begin{center} 
\resizebox{.75\textwidth}{!}{\includegraphics{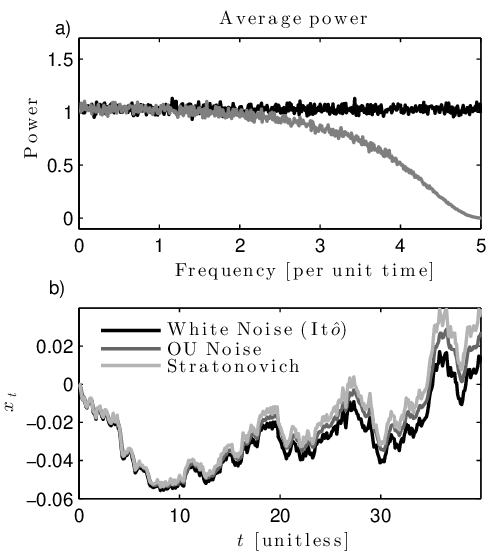}}
\end{center}
\caption{Comparison of white and Ornstein-Uhlenbeck noises. Plot a) is the average power, or modulus of the
 Fourier time series, of each noise. 
Plot b) is a realization of the solution to  system  \eqref{SDDE} with $\delta = 0$, $f(x) = a x$, $\sigma(x) = (bx+c)$, and the initial condition $x_0 = 0$, where $\xi_t$ is 
each type of noise.}
\label{fig:Compare}
\end{figure}
In Figure~\ref{fig:Compare}, white noise is compared to OU noise. White noise (black) is generated by taking the differential of 
a standard Wiener process. Ornstein-Uhlenbeck noise (dark gray) is generated by using the SDE~\eqref{yequation} with $\tau = 5$
 (arbitrary time units). In panel a) the average power as a function of frequency is plotted for each process. This is the Fourier transform 
pair of the correlation function. For the white noise process the average power is constant as a function of frequency, while the average power of the OU process decays rapidly 
after frequency $2$. In panel b), realizations are shown of a one-dimensional example of system~\eqref{SDDE}, with no delay, driven by the noises. 
Notice that the process driven by the OU noise and that driven by the white noise diverge from each other after time $t=20$.  This is a result of the Stratonovich correction that appears in \eqref{limiting equation} when one sets the delay equal to zero.


Theorem~\ref{theorem}, as stated, only applies to the case where the noise is modeled by a stationary OU process. The theorem can be 
modified to handle any noise process that is stationary and that solves a linear SDE with additive (white) noise (e.g. the harmonic noise process in \cite{mcdaniel2016,pesce2013}).  For different noises, the coefficients of the additional drift in the analogous limiting equation depend on the covariance function of the noise.  We expect that, with extra work, stationary noises defined by other SDEs may also be treated.
For a general noise 
$\bm{\xi}$, defined by its covariance function, different methods must be employed to prove 
an analogous theorem.

\section{Proof of Theorem~\ref{theorem}}\label{sec:proof of theorem}

In this section we prove the main theorem of the paper, Theorem~\ref{theorem}.  The main tool that we use is a theory of convergence of solutions of SDEs by Kurtz and Protter \cite{kurtz91}.  A similar technique is used in \cite{hottovy2014} and in \cite{mcdaniel2016} .  The structure of the section is as follows. In Section~\ref{sec:convergenceSI} we introduce the theory of convergence of solutions of SDEs that we will use to prove the theorem.  In Section~\ref{sec:Derive} we use integration by parts and
substitution to write system~(\ref{SDDE}) in the form necessary for applying the Kurtz-Protter theorem. In Section~\ref{sec:Proof} we 
complete the proof of Theorem~\ref{theorem} by verifying that the conditions of the Kurtz-Protter theorem are satisfied. 

We begin with the theory of convergence of solutions of SDEs, where we state (in a less general but sufficient form) a theorem of Kurtz and Protter \cite{kurtz91}.

\subsection{Convergence of solutions of SDEs}\label{sec:convergenceSI}

We fix a probability space $(\Omega,\mathcal{F},P)$ and an $n$-dimensional Wiener process $\bm{W}$ on it.  Let $\mathcal{F}_t$ will be (the usual augmentation of) $\sigma(\{\bm{W}_s:s\leq t\})$, the filtration generated by $W$ up to time $t$ \cite{revuz}.

Suppose $\bm{H}$ is an $\{\mathcal{F}_t\}$-adapted semimartingale with paths in $C([0,T] , \mathbb{R}^k)$, whose Doob-Meyer decomposition is $\bm{H}_t = \bm{M}_t + \bm{A}_t$ so that $\bm{M}$ is an $\mathcal{F}_t$-local martingale and $\bm{A}$ is a process of locally bounded variation, such that $\bm{A}_0 = 0$ \cite{revuz}.  For a continuous $\{\mathcal{F}_t\}$-adapted process $\bm{Y}$ with paths in $C([0,T] , \mathbb{R}^{d \times k})$ and for $t \leq T$ consider the It\^o integral 
\begin{equation}\label{eq:defint}
\int_0^t \bm{Y}_s\,d\bm{H}_s = \lim \sum_{i}\bm{Y}_{t_i} \left( \bm{H}_{t_{i+1}} - \bm{H}_{t_i} \right) \; ,
\end{equation}
where $\{t_i\}$ is a partition of $[0,t]$ and the limit is taken as the maximum of $t_{i+1}-t_i$ goes to zero. For a continuous processes $\bm{Y}$ such that 
\begin{equation*}
P\left[ \int_0^T \left\| \bm{Y}_s \right\| ^2 \,d\langle \bm{M} \rangle_s + \int_0^T \| \bm{Y}_s \| \,dV_s(\bm{A}) <\infty\right] = 1 \; , 
\end{equation*}
where $\langle \bm{M} \rangle_s$ is the quadratic variation of $\bm{M}$ and $V_s(\bm{A})$ is the total variation of $\bm{A}$, the limit in equation~(\ref{eq:defint}) exists in the sense that 
$$
\sup_{0 \leq t \leq T} \left \| \int_0^t \bm{Y}_s\,d\bm{H}_s - \sum_{i}\bm{Y}_{t_i} \left( \bm{H}_{t_{i+1}} - \bm{H}_{t_i} \right) \right \| \rightarrow 0 
$$
in probability.  This and related convergence modes will be used throughout the paper \cite{protter}.

Consider $(\bm{U}^\epsilon,\bm{H}^\epsilon)$ with paths in $C([0,T], \mathbb{R}^m\times\mathbb{R}^k)$ adapted to $\{\mathcal{F}_t\}$ where $\bm{H}^\epsilon$ is a semimartingale with respect to $\mathcal{F}_t$. Let $\bm{H}^\epsilon_t = \bm{M}_t^\epsilon + \bm{A}_t^\epsilon$ be its Doob-Meyer decomposition.  Let $\bm{h}:\mathbb{R}^{k} \rightarrow \mathbb{R}^{k \times n}$ be a {matrix-valued} function and let $\bm{X}^\epsilon$, with paths in $C([0,T], \mathbb{R}^m)$, satisfy the SDE 
\begin{equation}\label{eq:KPthm1}
\bm{X}_{t}^{\epsilon} = \bm{X}_0 + \bm{U}_{t}^\epsilon + \int_0^{t} \bm{h}(\bm{X}_s^{\epsilon})\,d\bm{H}_s^\epsilon \; ,
\end{equation}
where $\bm{X}_0^{\epsilon} = \bm{X}_0 \in\mathbb{R}^m$ is the same initial condition for all $\epsilon$. Define $\bm{X}$, with paths in $C([0,T], \mathbb{R}^m)$, to be the solution of
\begin{equation}\label{eq:KPlimit}
\bm{X}_t = \bm{X}_0 + \int_0^t \bm{h}(\bm{X}_s)\,d\bm{H}_s \; .
\end{equation}
Note that (\ref{eq:KPthm1}) implies $\bm{U}_0^\epsilon = \bm{0}$ for all $\epsilon$. 

\begin{lemma}\cite[Theorem 5.4 and Corollary 5.6]{kurtz91}
\label{theorem:KP} 
Suppose $(\bm{U}^\epsilon,\bm{H}^\epsilon)\rightarrow (\bm{0},\bm{H})$ in probability with respect to $C([0,T], \mathbb{R}^m\times\mathbb{R}^k)$, i.e., for all $a > 0$, 
\begin{equation}\label{eq:defofprob}
P \left[ \sup_{0\leq s \leq T} \Big( \| \bm{U}_s^\epsilon \| + \| \bm{H}_s^\epsilon-\bm{H}_s \| \Big) > a \right] \rightarrow 0 
\end{equation}
as $\epsilon \rightarrow 0$, and the following conditions are satisfied:

\begin{condition}\label{condition}[Tightness condition]
The family of total variations evaluated at $T$, $\{V_T(\bm{A}^\epsilon)\}$, is stochastically bounded, i.e., $P[V_T(\bm{A}^\epsilon)>L]\rightarrow 0$ as $L\rightarrow\infty$, uniformly in $\epsilon$.
\end{condition}

\begin{condition}\label{condition2}
$\bm{h}$ is continuous (see \cite[Example 5.3]{kurtz91}).
\end{condition}

Suppose that there exists a strongly unique global solution to equation~(\ref{eq:KPlimit}). Then, as $\epsilon\rightarrow0$, $\bm{X}^\epsilon$ converges to $\bm{X}$, the solution of equation~(\ref{eq:KPlimit}), in probability with respect to $C([0,T] , \mathbb{R}^m)$.
\end{lemma}

\subsection{Derivation of the limiting equation}
\label{sec:Derive}

\begin{proof}[Proof of Theorem \ref{theorem}]
To write system (\ref{SDDE}) (with $\bm{\xi} _t = \bm{y} _t ^{\epsilon}$) in the form of (\ref{eq:KPthm1}), we will use integration by parts 
and substitution. To this end, 
we write equation (\ref{yequation}) in matrix form.  Recalling that $\tau _j = k_j \epsilon$ and defining
\begin{equation*} \bm{D} =
\left[\rule{0cm}{1.4cm}\right. \begin{array}{cccc}
\frac{1}{k_1} & 0 & ... & 0  \\[3pt]
0 & \frac{1}{k_2} & ... & 0\\
\vdots & \vdots & \ddots & \vdots \\[3pt]
0 & 0 & ... & \frac{1}{k_n}
\end{array} \left.\rule{0cm}{1.4cm}\right] \; ,
\end{equation*}
equation (\ref{yequation}) becomes
\begin{equation}
\label{ymatrixform}
d \bm{y} ^\epsilon _t = \frac{\bm{D}}{\epsilon} \Big( - \bm{y} ^{\epsilon} _t dt + d \bm{W} _t \Big) \; .
\end{equation}
We solve equation (\ref{ymatrixform}) for $\bm{y}^{\epsilon} _t dt$ and substitute it into equation (\ref{SDDE}) to obtain
$$d \bm{x} ^{\epsilon} _t = \bm{f}(\bm{x} ^{\epsilon} _t) dt + \bm{\sigma}(\bm{x} ^{\epsilon}_{t - \bm{c} \epsilon}) d \bm{W} _t - \bm{\sigma}(\bm{x} ^{\epsilon}_{t - \bm{c} \epsilon }) \epsilon \bm{D} ^{-1} d \bm{y} ^{\epsilon} _t$$
where $\bm{x} ^{\epsilon}_{t - \bm{c} \epsilon} = \big( (x ^{\epsilon}_{t - c_1 \epsilon})_1 , ..., (x ^{\epsilon}_{t - c_i \epsilon})_i , ..., (x ^{\epsilon}_{t - c_m \epsilon})_m \big) ^T$.  In integral form, this equation becomes
\begin{equation}
\label{eq:PreParts}
\bm{x} ^{\epsilon} _t = \bm{x} _0 + \int_0 ^t \bm{f}(\bm{x} ^{\epsilon} _s) ds + \int_0 ^t \bm{\sigma}(\bm{x} ^{\epsilon} _{s - \bm{c} \epsilon}) d\bm{W} _s - \int_0 ^t \bm{\sigma}(\bm{x} ^{\epsilon} _{s - \bm{c} \epsilon}) \epsilon \bm{D} ^{-1} d \bm{y} ^{\epsilon} _s \; .
\end{equation}

In the limit as $\epsilon\rightarrow0$, we expect the second and the third terms on the right-hand side of equation (\ref{eq:PreParts})
to converge to the analogous terms of the limiting equation (\ref{limiting equation}). In addition, 
$\epsilon \bm{y}^\epsilon$ goes to zero as $\epsilon\rightarrow 0$, as we will show later (see Lemma~\ref{lemma:stochasticconvolution}).
Because of this, one might expect  the last term of the right-hand side to converge to zero as well. This is the case when $\bm{\sigma}(\bm{x})=\bm{\sigma}$ is a constant function. For non-constant $\bm{\sigma}(\bm{x})$, in order to be able to apply Lemma \ref{theorem:KP} directly  the process 
$\epsilon \bm{y} ^\epsilon$ would have to satisfy Condition \ref{condition}.  This is not the case, nor is it true 
for any family of colored noise
processes which converge to white noise (see \cite{kurtz91}). 

To resolve this issue, we first split the last integral into a part that involves values of the past condition and a part that does not.
We then integrate by parts the $i^{\mathrm{th}}$ component of the latter integral to obtain
\begin{align*}
(x ^{\epsilon} _t)_i = ( & x _0)_i + \int_0 ^t f_i (\bm{x} ^{\epsilon} _s) ds + \left( \int_0 ^t \bm{\sigma}(\bm{x} ^{\epsilon} _{s - \bm{c} \epsilon }) d \bm{W} _s \right) _i \\
& - \left( \int_0 ^{c^* \epsilon} \bm{\sigma}(\bm{x} ^{\epsilon} _{s - \bm{c} \epsilon}) \epsilon \bm{D} ^{-1} d \bm{y} ^{\epsilon} _s \right) _i 
  - \sum _j \sigma_{ij} (\bm{x} ^{\epsilon} _{t - \bm{c} \epsilon }) k_j \epsilon (y ^{\epsilon} _t)_j \\
 &+ \sum_j \sigma_{ij}(\bm{x} ^{\epsilon} _{ c^* \epsilon - \bm{c} \epsilon }) k _j \epsilon (y_{c^* \epsilon} ^{\epsilon})_j + \int_{c^* \epsilon}^t \sum _{p,j} \frac{\partial \sigma_{ij}}{\partial x_p} (\bm{x} ^{\epsilon} _{s - \bm{c} \epsilon }) k _j \epsilon (y ^{\epsilon} _s)_j d (x ^{\epsilon} _{s - c_p \epsilon })_p \; .
\end{align*}
Recall that $c^* \epsilon$ denotes the largest delay.  We note that the last term is a Lebesgue integral since $d(x_{s - c_p \epsilon}^{\epsilon})_p$ is the differential of a continuous process with bounded variation.  We substitute equation (\ref{SDDE}) into the term containing $d(x ^{\epsilon}_{s - c_p \epsilon })_p $ to obtain
\begin{align}
\label{xequationaftersub}
(x ^{\epsilon} _t)_i = (x_0)_i &+ \int_0 ^t f_i(\bm{x} ^{\epsilon} _s) ds + \left( \int_0 ^t \bm{\sigma}(\bm{x} ^{\epsilon} _{s - \bm{c} \epsilon }) d \bm{W} _s \right) _i \hspace{130pt} \notag \\[2pt]
&+ \int_{c^* \epsilon}^t \sum _{p, j} \frac{\partial \sigma_{ij}}{\partial x_p}(\bm{x} ^{\epsilon} _{s - \bm{c} \epsilon }) k _j \epsilon (y ^{\epsilon} _s)_j f_p (\bm{x} ^{\epsilon} _{s - c_p \epsilon }) ds \notag \\[2pt]
&+ \int_{c^* \epsilon}^t \sum_{p, j, \ell} \frac{\partial \sigma_{ij} }{\partial x_p} (\bm{x} ^{\epsilon} _{s - \bm{c} \epsilon }) k _j \epsilon (y ^{\epsilon} _s)_j \sigma_{p \ell} (\bm{x} ^{\epsilon} _{s - c_p \epsilon - \bm{c} \epsilon }) (y ^{\epsilon}_{s - c_p \epsilon })_{\ell} ds \notag \\[3pt]
& - \left( \int_0 ^{c^* \epsilon} \bm{\sigma}(\bm{x} ^{\epsilon} _{s - \bm{c} \epsilon}) \epsilon \bm{D} ^{-1} d \bm{y} ^{\epsilon} _s \right) _i \notag \\[2pt]
& - \sum _j \sigma_{ij} (\bm{x} ^{\epsilon} _{t - \bm{c} \epsilon }) k_j \epsilon (y ^{\epsilon} _t)_j + \sum _j \sigma_{ij} (\bm{x} ^{\epsilon} _{ c^* \epsilon - \bm{c} \epsilon }) k_j \epsilon (y_{c^* \epsilon} ^{\epsilon})_j
\end{align}
where $\bm{x} ^{\epsilon} _{s - c_p \epsilon } = \big( (x ^{\epsilon}_{s - c_p \epsilon })_1  , ... , (x ^{\epsilon}_{s - c_p \epsilon })_i  , ... , (x ^{\epsilon} _{s - c_p \epsilon })_m \big) ^{\rm T}$ and \\ $\bm{x} ^{\epsilon} _{s - c_p \epsilon - \bm{c} \epsilon } = \big( (x ^{\epsilon}_{s - c_p \epsilon - c_1 \epsilon } )_1 , ... , (x ^{\epsilon}_{s - c_p \epsilon - c_i \epsilon })_i  , ... , (x ^{\epsilon}_{s - c_p \epsilon - c_m \epsilon })_m \big) ^{\rm T}$.

\vspace{5pt}

In the above equation, we will show that the boundary terms and the second Lebesgue integral (the one whose integrand contains the factor $f_p (\bm{x} ^{\epsilon} _{s - c_p \epsilon })$)
go to zero because $\epsilon \bm{y}^\epsilon$ goes to zero. To prove convergence, 
we add and subtract like terms without delays on the right-hand side of the above equation. In addition, 
we add and subtract the integral of $S _i (\bm{x})$, where $S _i (\bm{x})$ is defined in equation~(\ref{eq:Sdef}).  In the resulting equation for $\bm{x}^{\epsilon}$, we keep with the notation of Lemma \ref{theorem:KP} by collecting the terms which we will show directly go to zero into a new process called $\bm{U}^{\epsilon}$.  At this point, it is convenient to extend the process $\bm{y} ^{\epsilon}$ so that $\bm{y} ^{\epsilon} _s$ is defined for $s \geq t_-$, where $t_- < 0$ was introduced before 
the statement of Theorem~\ref{theorem}.  We do this in such a way so that 
$\bm{y} ^{\epsilon}$ considered on the interval $s \geq t_-$ is a stationary process.

Let $\bm{U}^{\epsilon}$ be defined componentwise as
\begin{align}
\label{eq:U}
(U^{\epsilon} _t)_i = \; & - \sum _j \sigma_{ij} (\bm{x} ^{\epsilon} _{t - \bm{c} \epsilon }) k_j \epsilon (y ^{\epsilon} _t)_j + \sum _j \sigma_{ij} (\bm{x} ^{\epsilon} _{ c^* \epsilon - \bm{c} \epsilon }) k_j \epsilon (y_{c^* \epsilon} ^{\epsilon})_j \notag \hspace{110pt} \\
& - \left( \int_0 ^{c^* \epsilon} \bm{\sigma}(\bm{x} ^{\epsilon} _{s - \bm{c} \epsilon}) \epsilon \bm{D} ^{-1} d \bm{y} ^{\epsilon} _s \right) _i \notag \\
& + \left( \int_0 ^t \Big( \bm{\sigma}(\bm{x} ^{\epsilon} _{s - \bm{c} \epsilon }) - \bm{\sigma}(\bm{x}_{s}^\epsilon) \Big)  d \bm{W} _s \right) _i \notag \\[2pt]
&+ \int_{c^* \epsilon}^t \sum _{p, j} \frac{\partial \sigma_{ij}}{\partial x_p} (\bm{x} ^{\epsilon} _{s - \bm{c} \epsilon }) k _j \epsilon (y ^{\epsilon} _s)_j f_p (\bm{x} ^{\epsilon} _{s - c_p \epsilon }) ds \\[1pt]
&+ \int_{c^* \epsilon}^t \sum_{p, j, \ell} \Bigg[ \frac{\partial \sigma_{ij}}{\partial x_p} (\bm{x} ^{\epsilon} _{s - \bm{c} \epsilon }) \sigma_{p \ell} (\bm{x} ^{\epsilon} _{s - c_p \epsilon - \bm{c} \epsilon }) \notag \\
& \hspace{110pt} - \frac{\partial \sigma_{ij}}{\partial x_p} (\bm{x} ^{\epsilon} _{s}) \sigma_{p \ell} (\bm{x} ^{\epsilon} _{s}) \Bigg] k _j \epsilon (y ^{\epsilon} _s)_j (y ^{\epsilon}_{s - c_p \epsilon })_{\ell}  ds \notag \\
&- \int_0^{c^* \epsilon} \sum_{p, j, \ell} \frac{\partial \sigma_{ij}}{\partial x_p} (\bm{x} ^{\epsilon} _s) \sigma_{p \ell} (\bm{x} ^{\epsilon} _s ) k _j \epsilon (y ^{\epsilon} _s)_j (y ^{\epsilon}_{s - c_p \epsilon })_{\ell}  ds \; . \notag 
\end{align}
Then $\bm{x}^\epsilon_t$ can be written componentwise as
\begin{align}
\label{xepsilon equation}
(x ^{\epsilon} _t)_i = \; & (x_0)_i + (U^{\epsilon} _t)_i + \int_0 ^t f_i(\bm{x} ^{\epsilon} _s) ds + \left( \int_0 ^t \bm{\sigma}(\bm{x} ^{\epsilon} _{s}) d \bm{W} _s \right) _i \notag \hspace{105pt} \\[3pt]
&+ \int_0^t \sum_{p, j} \frac{1}{2} e^{- \frac{c_p}{k_j}} \frac{\partial \sigma_{ij}}{\partial x_p}  (\bm{x} ^{\epsilon} _{s}) \sigma_{pj} (\bm{x} ^{\epsilon} _{s})  ds \notag \\[3pt]
&+ \int_0^t \sum_{p, j}  \frac{\partial \sigma_{ij}}{\partial x_p} (\bm{x} ^{\epsilon} _{s}) \sigma_{pj} (\bm{x} ^{\epsilon} _{s}) \left( k _j \epsilon (y ^{\epsilon} _s)_j (y ^{\epsilon} _{s - c_p \epsilon })_j - \frac{1}{2} e^{- \frac{c_p}{k_j}} \right) ds \notag \\[3pt]
&+ \int_0^t \sum_{\substack{p, j, \ell \\ j \neq \ell}}  \frac{\partial \sigma_{ij}}{\partial x_p} (\bm{x} ^{\epsilon} _{s}) \sigma_{p \ell} (\bm{x} ^{\epsilon} _{s}) k _j \epsilon (y ^{\epsilon} _s)_j (y ^{\epsilon}_{s - c_p \epsilon })_{\ell}   ds \; .
\end{align}
We now write equation (\ref{xepsilon equation}) in the form (\ref{eq:KPthm1}) of Lemma \ref{theorem:KP} by letting $\bm{h} : \mathbb{R} ^m \rightarrow \mathbb{R}^{m \times (1 + n + 1 +mn^2)}$ be the matrix-valued function given by
\begin{align} 
\label{hdefinition}
\bm{h} (\bm{X}) = \Big( \bm{f}(\bm{X})&, \bm{\sigma}(\bm{X}), \bm{S}(\bm{X}), \bm{\Lambda}^{11} (\bm{X}) , ... , \bm{\Lambda} ^{1n} (\bm{X}), \notag \hspace{135pt} \\
& \hspace{50pt} \bm{\Lambda}^{21} (\bm{X}) , ..., \bm{\Lambda}^{2n} (\bm{X}), \bm{\Lambda}^{n1} (\bm{X}), ..., \bm{\Lambda}^{nn} (\bm{X}) \Big)
\end{align}
where $\bm{S} : \mathbb{R} ^m \rightarrow \mathbb{R}^{m}$ is the vector-valued function defined in equation~(\ref{eq:Sdef})
and $\bm{\Lambda}^{j \ell} : \mathbb{R} ^m \rightarrow \mathbb{R}^{m \times m}$ is defined componentwise as
\begin{equation*}
\Lambda ^{j \ell} _{i p} (\bm{X}) = \frac{\partial \sigma_{ij}}{\partial X_p} (\bm{X}) \sigma_{p \ell} (\bm{X})
\end{equation*}
and by letting $\bm{H}^\epsilon$ be the process, with paths in $C \big( [0,T], \mathbb{R}^{(1 + n + 1 + mn^2)} \big)$, given by
\begin{equation}
\label{eq:DefH}
\bm{H} _t^\epsilon =
\left[\begin{array}{c}
t  \\
\bm{W}_t \\
t  \\
\bm{G}^{11} _t \\
\vdots \\
\bm{G}^{nn} _t
\end{array}\right] ,
\end{equation}
where $\bm{G}^{j \ell}$ is the process, with paths in $C([0,T] , \mathbb{R}^m)$, given by
\begin{equation}
\label{eq:DefG}
\big( G ^{j \ell} _t \big) _p = \int _0 ^t \bigg( k _j \epsilon ({y} ^{\epsilon} _s)_j ({y} ^{\epsilon}_{s - c_p \epsilon })_{\ell}  -E \Big[ k_j\epsilon(y_s^\epsilon)_j ({y} ^{\epsilon}_{s - c_p \epsilon })_{\ell} \Big] \bigg)ds \; . 
\end{equation}
Note that the expectation in the integrand in (\ref{eq:DefG}) is equal to 
$\frac{1}{2}e^{- c_p / k_j }$ for $j=\ell$ and equal to zero for $j\neq \ell$.

We will show that the processes $\bm{G} ^{j\ell}$ converge to zero in $L^2$ with respect to $C([0,T] , \mathbb{R}^m)$ for all 
$j,\ell=1,...,n$  (Lemma~\ref{lemma:uniformaverage}). Thus, the limiting process $\bm{H}$ is given by
\begin{equation}
	\bm{H}_t = \left[\begin{array}{c}
t  \\
\bm{W}_t \\
t  \\
\bm{0} \\
\vdots \\
\bm{0}
\end{array}\right] .
\end{equation}
We show in the next subsection that $\bm{U}^{\epsilon}, \bm{H}^{\epsilon}, \bm{H},$ and $\bm{h}$ satisfy the assumptions of Lemma~\ref{theorem:KP}.  Thus, given the definitions of $\bm{h}$ and $\bm{H}$ above, by computing the right-hand side of
\begin{equation}
\bm{x}_t = \bm{x}_0 + \int_0^t \bm{h}(\bm{x}_s)\;d\bm{H}_s
\end{equation}
we get the limiting SDE (\ref{limiting equation}).

\subsection{Verifying Conditions of Lemma~\ref{theorem:KP}}

\label{sec:Proof}
In this subsection, we complete the proof of Theorem \ref{theorem} by showing that the conditions of Lemma~\ref{theorem:KP} are satisfied.  That is, we show $\bm{U}^\epsilon\rightarrow \bm{0}$ in probability with respect to $C([0,T], \mathbb{R}^m)$, $\bm{H}^\epsilon\rightarrow \bm{H}$ in probability with respect to  $C([0,T], \mathbb{R} ^{(1 + n + 1 + m n^2)})$, and Conditions \ref{condition} and \ref{condition2} are satisfied.
We begin with lemmas that we will need.  First we show that $\epsilon \bm{y}^\epsilon\rightarrow \bm{0}$ in $L^2$ with respect to $C([0,T] , \mathbb{R}^n) $. 
Nelson showed a similar result, namely that $\sup_{0\leq t\leq T}\|\epsilon \bm{y}_t^\epsilon\| \rightarrow 0$ as $\epsilon\rightarrow0$ with probability one \cite{nelson}. 

\begin{lemma}
\label{lemma:stochasticconvolution}
For each $\epsilon > 0$, let $\bm{y}^\epsilon$ be the solution to equation~(\ref{yequation}) with $\bm{y}^\epsilon _0$ distributed according to the stationary distribution corresponding to (\ref{yequation}).  There exists $C'>0$ independent of $\epsilon$ such that 
\begin{equation}
\label{eq:Conv}
 E \left[ \left( \sup_{0 \leq t \leq T} \left| \int _0 ^t e^{-\frac{(t-s)}{k_j \epsilon}} \; d(W_s)_j \right| \right) ^2 \right] \leq C' \epsilon ^{1 / 2} \; ,
\end{equation}
and this implies that there exists $C > 0$ independent of $\epsilon$ such that
\begin{equation}
\label{ybound}
E\left [\sup_{0\leq t\leq T} \left\| \epsilon \bm{y}_t^\epsilon \right\| ^2\right ] \leq C \epsilon ^{1 / 2} \; .
\end{equation}
Thus, as
$\epsilon\rightarrow0$, $\epsilon \bm{y}^\epsilon\rightarrow \bm{0}$ in $L^2$, and therefore in probability, with respect
to $C([0,T] , \mathbb{R}^n)$.
\end{lemma}

\begin{proof}
We begin by proving (\ref{eq:Conv}) by following the first part of the argument of \cite[Lemma 3.7]{pavliotis2005}.  We fix $\alpha \in (0, \frac{1}{2})$ and use the factorization method from \cite{daprato1988} (see also \cite[Section 5.3]{daprato}) to rewrite
\begin{align*}
I(t) &= \int _0 ^t e^{-\frac{(t-s)}{k_j \epsilon}} \; d(W_s)_j \\[3pt]
 &= \frac{\sin(\pi \alpha)}{\pi} \int _0 ^t  e^{-\frac{(t-s)}{k_j\epsilon}} (t - s)^{\alpha - 1} Y(s) ds \; ,
\end{align*}
where 
$$Y(s) = \int_0^s e^{-\frac{(s-u)}{k_j \epsilon}} (s - u)^{- \alpha} d(W_u)_j $$
and we used the identity
$$\int _u ^t (t - s) ^{\alpha - 1} (s - u)^{- \alpha} ds = \frac{\pi}{\sin(\pi \alpha)} \; , \hspace{10pt} 0 < \alpha < 1 \; .$$

We fix $m > \frac{1}{2 \alpha}$ and use the H$\ddot{\mathrm{o}}$lder inequality:
\begin{equation*}
| I(t) | ^{2m} \leq C_1 \left( \int _0 ^t \left| e^{-\frac{(t-s)}{k_j \epsilon}} (t - s)^{\alpha - 1} \right| ^{\frac{2m}{2m - 1}} ds \right) ^{2m - 1} \int_0 ^t | Y(s) | ^{2m} ds \; .
\end{equation*}
Using the change of variables $z = \frac{2m}{2m-1} \frac{(t-s)}{k_j \epsilon} $, we have
\begin{align*}
\int _0 ^t e^{- \frac{2m}{2m - 1} \frac{(t - s)}{k_j \epsilon}} (t - & s)^{\frac{2m}{2m - 1} ( \alpha - 1)} ds  \\
& = \left(\frac{k_j (2m - 1)}{2 m} \right) ^{\frac{2m \alpha - 1}{2m - 1}} \epsilon ^{\frac{2m \alpha - 1}{2m - 1}} \int_0 ^{\frac{t}{k_j \epsilon} \frac{2 m}{2m - 1}} e ^{- z} z ^{\frac{2m}{2m - 1} (\alpha - 1)} dz \\[5pt]
& \leq C_2 \epsilon ^{\frac{2m \alpha - 1}{2m - 1}} \int_0 ^{\infty} e ^{- z} z ^{\frac{2m}{2m - 1} (\alpha - 1)} dz \\[4pt]
& \leq C_3 \epsilon ^{\frac{2m \alpha - 1}{2m - 1}}
\end{align*}
where in the above we have used the fact that $e ^{- z} z ^{\frac{2m}{2m - 1} (\alpha - 1)} \in L^1 (\mathbb{R} ^+)$ since $m > \frac{1}{2 \alpha}$.  Therefore, we have
\begin{equation*}
E \left[ \sup_{0 \leq t \leq T} | I(t) | ^{2m} \right] \leq C_4 \epsilon ^{2m \alpha - 1} E \left[ \int _0^T | Y(s) | ^{2m} ds \right] .
\end{equation*}

Then, letting $\frac{1}{4} < \alpha < \frac{1}{2}$ and $m = 2$, we have
\begin{align*}
E \Big[ \big| Y(t) \big| ^{4} \Big] & = 3 \bigg( E \left[ \big( Y(t) \big) ^2  \right] \bigg) ^2 \hspace{20pt} \textrm{since $Y(t)$ is a zero-mean Gaussian}\\[3pt]
& = 3 \Bigg( \int_0^t  e^{-\frac{2(t-u)}{k_j \epsilon}} (t - u)^{- 2 \alpha}  \; du \Bigg) ^2 \hspace{30pt} \textrm{by the It\^o isometry} \\[3pt]
& = 3 \left( \frac{k_j \epsilon}{2} \right) ^{2(1 - 2 \alpha)} \left( \int_0 ^{\frac{2t}{k_j \epsilon}} e^{-s} s^{-2 \alpha} ds \right) ^2 \\[2pt]
& \leq C_5 \epsilon ^{2(1 - 2 \alpha)} \hspace{11pt} \textrm{using the fact that $e^{-s} s^{-2 \alpha} \in L^1 (\mathbb{R} ^+)$ since $\alpha < \frac{1}{2}$} \; .  
\end{align*}
Thus,
\begin{equation*}
E \left[ \sup_{0 \leq t \leq T} | I(t) | ^{4} \right] \leq C_6 \epsilon \; , 
\end{equation*}
where $C_6$ is a constant that depends on $T$, so we get (\ref{eq:Conv}) by the Cauchy-Schwarz inequality:
\begin{equation*}
E \left[ \sup_{0 \leq t \leq T} | I(t) | ^{2} \right] \leq \Bigg( E \left[ \sup_{0 \leq t \leq T} | I(t) | ^{4} \right] \Bigg) ^{1 / 2} \leq C' \epsilon ^{1 / 2} \; . 
\end{equation*}

Now we prove (\ref{ybound}).  The solution of (\ref{yequation}) is
\begin{equation*}
	(y_t^\epsilon)_j = e^{-\frac{t}{k_j \epsilon}}(y_0^\epsilon)_j +\frac{1}{k_j \epsilon} \int_0^t e^{-\frac{(t-s)}{k_j \epsilon}}\;
	d(W_s)_j \; . 
\end{equation*}
Thus, we have
\begin{align*}
	E\left [\sup_{0\leq t\leq T} \big| (\epsilon y_t^\epsilon)_j \big| ^2 \right ] \leq & \; 2E\left [\sup_{0\leq t\leq T} \bigg| e^{-\frac{t}{k_j \epsilon}}(\epsilon y_0 ^{\epsilon})_j \bigg| ^2 \right ] \hspace{120pt} \\
  & \hspace{50pt} + \frac{2}{k_j ^2} E\left[ \sup_{0\leq t\leq T} \Bigg( \int_0^t e^{-\frac{(t-s)}{k_j \epsilon}}\;d(W_s)_j \Bigg) ^2 \right] \; .
\end{align*}
For the first term we use that $(y_0^\epsilon)_j$ is distributed according to the stationary distribution corresponding to (\ref{yequation}) and thus has mean zero and variance $E \big[ |(y_0^\epsilon)_j|^2 \big] = (2k_j\epsilon)^{-1}$:
\begin{equation*}
E\left [\sup_{0\leq t\leq T} \bigg| e^{-\frac{t}{k_j \epsilon}}(\epsilon y_0 ^{\epsilon})_j \bigg| ^2 \right ] \leq E\left[ \big| \epsilon (y_0^\epsilon)_j \big| ^2 \right] = \frac{\epsilon}{2 k_j} \; .
\end{equation*}
The bound (\ref{ybound}) then follows from this bound together with (\ref{eq:Conv}). 
\end{proof}

The next lemma is elementary but we include its proof for completeness. 

\begin{lemma}
	\label{lemma:ey2toConstant}
For each $\epsilon > 0$, let $\bm{y}^\epsilon$ be defined on $t \geq 0$ as the solution to equation~(\ref{yequation}) with $\bm{y}^\epsilon _0$ distributed according to the stationary distribution corresponding to (\ref{yequation}), and let $\bm{y}^\epsilon$ be defined on $t_- \leq t < 0$ in such a way so that $\bm{y}^\epsilon$ considered on the interval $t \geq t_-$ is a stationary process.  Then there exists $C > 0$ such that for all $\epsilon > 0$,
$$E \left[ \Big| \epsilon( y_t ^{\epsilon})_j (y_s ^{\epsilon})_{\ell} \Big| ^2 \right] < C$$ 
for all $1 \leq j, \ell \leq n$ and $t,s \geq t_- \; .$ 
\end{lemma}

\begin{proof}

For all $t \geq t_-$, $(y_t^\epsilon)_j$ is a mean-zero Gaussian random variable with variance
\begin{equation}
	\label{eq:Var}
	 E \left[ \big|(y_t^\epsilon)_j \big| ^2 \right] = \frac{1}{2k_j \epsilon} \; ,
\end{equation}
and so
\begin{equation*}
	 E \left[ \big|(y_t^\epsilon)_j \big| ^4 \right] = \frac{3}{4 k_j ^2 \epsilon ^2} \; .
\end{equation*}
Thus, by the Cauchy-Schwarz inequality
\begin{align*}
E \left[ \Big| \epsilon( y_t ^{\epsilon})_j (y_s ^{\epsilon})_{\ell} \Big| ^2 \right] \; \leq \; \bigg( E \left[ \big( \sqrt{\epsilon}y_t^\epsilon \big) _j^4 \right] \bigg) ^{1/2} \bigg( E \left[ \big( \sqrt{\epsilon}y_s^\epsilon \big) _\ell^4  \right] \bigg) ^{1/2} \;
                       = \; \frac{3}{4 k_j k_\ell} \; \; .
\end{align*}
\end{proof}

\begin{lemma}
	\label{lemma:continuityofxepsilon}
For each $\epsilon > 0$, let $\bm{x}^{\epsilon}$ be defined as in the statement of
Theorem \ref{theorem} where $\bm{f}$ is bounded and $\bm{\sigma}$ is bounded with bounded derivatives.  Also, for $\bm{t} = (t_1, ... , t_m) \in [0,T] ^m$ and $\bm{h} = (h_1 , ... , h_m) \in [0,1] ^m$, let $\bm{x} ^{\epsilon} _{\bm{t}} = \big( (x ^{\epsilon} _{t_1})_1 , ..., (x ^{\epsilon} _{t_m})_m \big)$, let $\bm{x} ^{\epsilon} _{\bm{t} + \bm{h}} = \big( (x ^{\epsilon} _{t_1 + h_1})_1 , ..., (x ^{\epsilon} _{t_m + h_m})_m \big)$, and let $h_M = \max_{1 \leq i \leq m} h_i$.  Then there exists $C$ independent of $\epsilon$ such that for all $\bm{t} \in [0,T]^m$, $\bm{h} \in [0,1] ^m$, and $\epsilon > 0$,
\begin{equation*}
E \left[ \left\| \bm{x}^{\epsilon} _{\bm{t} + \bm{h}} - \bm{x}^{\epsilon} _{\bm{t}} \right\| ^2 \right] \leq C \left( h_M + \epsilon ^{1/2} \right) \; .
\end{equation*}
\end{lemma}

\begin{proof}
From equation (\ref{xequationaftersub}) and the Cauchy-Schwarz inequality, we have
\begin{align*}
E & \left[ \Big| \left( x^{\epsilon} _{t_i + h_i} \right) _i - \left( x ^{\epsilon} _{t_i} \right) _i \Big| ^2 \right] \hspace{270pt} \\
& \; \hspace{10pt} \leq 6 E \left[ \left( \int_{t_i} ^{t_i + h_i} f_i(\bm{x} ^{\epsilon} _s) ds \right) ^2 \right] + 6 E \left[  \left(\rule{0cm}{.7cm}\right. \left( \int_{t_i} ^{t_i + h_i} \bm{\sigma}(\bm{x} ^{\epsilon} _{s - \bm{c} \epsilon }) d \bm{W} _s \right) _i \left.\rule{0cm}{.7cm}\right) ^2 \right]\\
& \hspace{23pt} + 6  E \left[ \left( \int_{t_i}^{t_i + h_i} \sum _{p, j} \frac{\partial \sigma_{ij}}{\partial x_p} (\bm{x} ^{\epsilon} _{s - \bm{c} \epsilon }) k _j \epsilon (y ^{\epsilon} _s)_j f_p (\bm{x} ^{\epsilon} _{s - c_p \epsilon }) ds \right) ^2 \right] \\
& \hspace{23pt} + 6 E \left[ \left( \int_{t_i}^{t_i + h_i} \sum_{p, j, \ell} \frac{\partial \sigma_{ij}}{\partial x_p} (\bm{x} ^{\epsilon} _{s - \bm{c} \epsilon }) k _j \epsilon (y ^{\epsilon} _s)_j \sigma_{p \ell} (\bm{x} ^{\epsilon} _{s - c_p \epsilon - \bm{c} \epsilon }) (y ^{\epsilon}_{s - c_p \epsilon })_{\ell}  ds \right) ^2 \right]  \\
& \hspace{23pt} + 6 E \left[ \left( \sum _j \sigma_{ij} (\bm{x} ^{\epsilon} _{t_i + h_i - \bm{c} \epsilon }) k_j \epsilon (y ^{\epsilon} _{t_i + h_i})_j \right) ^2 \right] \\
& \hspace{23pt} + 
6 E \left[ \left( \sum _j \sigma_{ij} (\bm{x} ^{\epsilon} _{t_i - \bm{c} \epsilon }) k_j \epsilon (y ^{\epsilon} _{t_i})_j \right) ^2 \right] .
\end{align*}
Using the boundedness of $f_i$, the It\^o isometry, and the Cauchy-Schwarz inequality then gives
\begin{align*}
 E \bigg[ & \Big| \left( x ^{\epsilon} _{t_i + h_i} \right) _i - \left( x ^{\epsilon} _{t_i} \right) _i \Big| ^2 \bigg] \hspace{265pt} \\[5pt]
& \; \leq C_1 {h_i}^2 + 6 C_2 \int_{t_i} ^{t_i + h_i} E \left[ \left\| \bm{\sigma}(\bm{x} ^{\epsilon} _{s - \bm{c} \epsilon }) \right\| ^2 \right] ds \\[3pt]
& \hspace{10pt} + 6 h_i   E \left[ \int_{t_i}^{t_i + h_i} \left( \sum _{p, j} \frac{\partial \sigma_{ij}}{\partial x_p} (\bm{x} ^{\epsilon} _{s - \bm{c} \epsilon }) k _j \epsilon (y ^{\epsilon} _s)_j f_p (\bm{x} ^{\epsilon} _{s - c_p \epsilon } ) \right) ^2 ds  \right] \\[3pt]
& \hspace{10pt} + 6 h_i E \left[  \int_{t_i}^{t_i + h_i} \left( \sum_{p, j, \ell} \frac{\partial \sigma_{ij}}{\partial x_p} (\bm{x} ^{\epsilon} _{s - \bm{c} \epsilon }) k _j \epsilon (y ^{\epsilon} _s)_j \sigma_{p \ell} (\bm{x} ^{\epsilon} _{s - c_p \epsilon - \bm{c} \epsilon }) (y ^{\epsilon}_{s - c_p \epsilon })_{\ell} \right) ^2 ds \right]  \\[3pt]
& \hspace{10pt} + 6n E \left[ \sum _j \Big( \sigma_{ij} (\bm{x} ^{\epsilon} _{t_i + h_i - \bm{c} \epsilon }) k_j \epsilon (y ^{\epsilon} _{t_i + h_i})_j \Big) ^2 \right] + 
6n E \left[ \sum _j \Big( \sigma_{ij} (\bm{x} ^{\epsilon} _{t_i - \bm{c} \epsilon }) k_j \epsilon (y ^{\epsilon} _{t_i})_j \Big) ^2 \right] .
\end{align*}
Using the boundedness of $\bm{f}$ and the boundedness of $\bm{\sigma}$ and its derivatives, we get
\begin{align*}
 E & \left[ \Big| \left( x ^{\epsilon} _{t_i + h_i} \right) _i - \left( x ^{\epsilon} _{t_i} \right) _i \Big| ^2 \right] \\
  & \hspace{64pt} \leq C_1 {h_i}^2 + C_3 h_i + 6 C_4 h_i  \int_{t_i}^{t_i + h_i} E \left[ \sum _{j} \Big( \epsilon (y ^{\epsilon} _s)_j \Big) ^2 \right] ds \\[3pt]
& \hspace{75pt} + 6 C_5 h_i  \int_{t_i}^{t_i + h_i} E \left[ \sum_{ j, \ell} \left( \epsilon (y ^{\epsilon} _s)_j  (y ^{\epsilon}_{s - c_p \epsilon })_{\ell} \right) ^2 \right] ds  \\[3pt]
& \hspace{75pt} + 6n C_6 E \left[ \sum _j \Big( \epsilon (y ^{\epsilon} _{t_i + h_i})_j \Big) ^2 \right] + 6n C_6 E \left[ \sum _j \Big( \epsilon (y ^{\epsilon} _{t_i})_j \Big) ^2 \right] .
\end{align*}
Thus, using Lemma \ref{lemma:stochasticconvolution} and Lemma \ref{lemma:ey2toConstant},
\begin{align*}
& E \left[ \Big| \left( x ^{\epsilon} _{t_i + h_i} \right) _i - \left( x ^{\epsilon} _{t_i} \right) _i \Big| ^2 \right] \leq C_1 {h_i}^2 + C_3 h_i +6 C_7 \epsilon ^{1 / 2} {h_i} ^2 + 6 C_8 {h_i} ^2 + C_9 \epsilon ^{1 / 2}
\end{align*}
from which the statement follows.
\end{proof}

\subsubsection{$ \textbf{U} ^{\epsilon}$ converges to zero}\label{sec:Utozero}

Now we are ready to prove that $\bm{U}^\epsilon$ goes to zero in probability with respect to $C([0,T] , \mathbb{R}^m)$ as $\epsilon \rightarrow 0$.  To do this, we split $(U^\epsilon)_i$ into three types of terms. Recall that $(U ^\epsilon)_i$ is defined in equation~(\ref{eq:U}) as 
\begin{align*}
(U^{\epsilon} _t)_i = \hspace{2pt} & \underbrace{ - \int_0^{c^* \epsilon} \sum_{p, j, \ell} \frac{\partial \sigma_{ij}}{\partial x_p} (\bm{x} ^{\epsilon} _s) \sigma_{p \ell} (\bm{x} ^{\epsilon} _s ) k _j \epsilon (y ^{\epsilon} _s)_j (y ^{\epsilon}_{s - c_p \epsilon })_{\ell}  ds }_{I} \\[5pt]
& \underbrace{ - \left( \int_0 ^{c^* \epsilon} \bm{\sigma}(\bm{x} ^{\epsilon} _{s - \bm{c} \epsilon}) \epsilon \bm{D} ^{-1} d \bm{y} ^{\epsilon} _s \right) _i}_{I} \\[5pt]
 & + \underbrace{\left( \int_0 ^t \Big( \bm{\sigma}(\bm{x} ^{\epsilon} _{s - \bm{c} \epsilon }) - \bm{\sigma}(\bm{x}_{s}^\epsilon) \Big)  d \bm{W} _s \right) _i}_{II} \\[5pt]
& + \int_{c^* \epsilon}^t \sum_{p, j, \ell} \Bigg[ \frac{\partial \sigma_{ij}}{\partial x_p} (\bm{x} ^{\epsilon} _{s - \bm{c} \epsilon }) \sigma_{p \ell} (\bm{x} ^{\epsilon} _{s - c_p \epsilon - \bm{c} \epsilon }) \\
& \hspace{10pt} \underbrace{ \hspace{120pt} - \frac{\partial \sigma_{ij}}{\partial x_p} (\bm{x} ^{\epsilon} _{s}) \sigma_{p \ell} (\bm{x} ^{\epsilon} _{s}) \Bigg] k _j \epsilon (y ^{\epsilon} _s)_j (y ^{\epsilon}_{s - c_p \epsilon })_{\ell}  ds}_{II} \\[5pt]
& \underbrace{- \sum _j \sigma_{ij} (\bm{x} ^{\epsilon} _{t - \bm{c} \epsilon}) k_j \epsilon (y ^{\epsilon} _t)_j + \sum _j \sigma_{ij} (\bm{x} ^{\epsilon} _{ c^* \epsilon - \bm{c} \epsilon}) k_j \epsilon (y_{c^* \epsilon} ^{\epsilon})_j}_{III} \\[5pt]
& + \underbrace{\int_{c^* \epsilon}^t \sum _{p, j} \frac{\partial \sigma_{ij}}{\partial x_p} (\bm{x} ^{\epsilon} _{s - \bm{c} \epsilon }) k _j \epsilon (y ^{\epsilon} _s)_j f_p (\bm{x} ^{\epsilon} _{s - c_p \epsilon }) ds}_{III} \; .
\end{align*}
 We will prove that these terms all converge to zero as $\epsilon\rightarrow 0$.  The type $I$ terms converge to zero because the length of the interval of integration goes to zero. 
Convergence of the type $II$ terms to zero will be a consequence of Lemma \ref{lemma:continuityofxepsilon} and the Lipschitz continuity of $\bm{\sigma}$ and its first derivatives.  The proof of convergence of type $III$ terms to zero will follow from 
 $\epsilon \bm{y}^\epsilon \rightarrow \bm{0}$ as $\epsilon\rightarrow 0$.  Recall that for $p \geq 1$, convergence in $L^p$ with respect to $C([0,T] , \mathbb{R})$ implies convergence in probability with respect to $C([0,T] , \mathbb{R})$.  Thus, to show that these terms converge to zero in probability with respect to $C([0,T] , \mathbb{R})$, we show that the fourth term goes to zero in $L^1$ with respect to $C([0,T] , \mathbb{R})$ and the other terms go to zero in $L^2$ with respect to $C([0,T] , \mathbb{R})$. We note that it is possible to show that the fourth term goes to zero in $L^2$ with respect to $C([0,T] , \mathbb{R})$, but it would take slightly more work. 

We start with the first two terms of $(U^{\epsilon} _t)_i$ , i.e., the type $I$ terms.  For the first term, by two separate applications of 
the Cauchy-Schwarz inequality, the boundedness of $\bm{\sigma}$ and its first derivatives, and Lemma \ref{lemma:ey2toConstant}, we have
\begin{align*}
& E \left[\rule{0cm}{.8cm}\right. \left(\rule{0cm}{.7cm}\right. \sup_{0 \leq t \leq T} \Bigg|
\int_0^{c^* \epsilon} \sum_{p, j, \ell} \frac{\partial \sigma_{ij}}{\partial x_p} (\bm{x} ^{\epsilon} _s) \sigma_{p \ell} (\bm{x} ^{\epsilon} _s ) k _j \epsilon (y ^{\epsilon} _s)_j (y ^{\epsilon}_{s - c_p \epsilon })_{\ell}  ds
\Bigg| \left.\rule{0cm}{.7cm}\right) ^2 \left.\rule{0cm}{.8cm}\right] \hspace{50pt} \\[3pt]
  & \hspace{60pt} \leq C_1 \epsilon \sum_{p, j, \ell}
\int_0^{c^* \epsilon}  E \left[\rule{0cm}{.7cm}\right. \Bigg| \frac{\partial \sigma_{ij}}{\partial x_p} (\bm{x} ^{\epsilon} _s) \sigma_{p \ell} (\bm{x} ^{\epsilon} _s ) k _j \epsilon (y ^{\epsilon} _s)_j (y ^{\epsilon}_{s - c_p \epsilon })_{\ell} \Bigg| ^2
\left.\rule{0cm}{.7cm}\right] ds \\[3pt]
& \hspace{60pt} \leq C_2 \epsilon ^2 \; .
\end{align*}
For the second term, using (\ref{ymatrixform}) and two separate applications of 
the Cauchy-Schwarz inequality, we have
\begin{align*}
& E \left[\rule{0cm}{.8cm}\right. \left(\rule{0cm}{.7cm}\right. \sup_{0 \leq t \leq T} \Bigg|
  \left( \int_0 ^{c^* \epsilon} \bm{\sigma}(\bm{x} ^{\epsilon} _{s - \bm{c} \epsilon}) \epsilon \bm{D} ^{-1} d \bm{y} ^{\epsilon} _s \right) _i \Bigg| \left.\rule{0cm}{.7cm}\right) ^2 \left.\rule{0cm}{.8cm}\right] \hspace{140pt} \\[3pt]
& \hspace{40pt} = E \left[\rule{0cm}{.7cm}\right. \Bigg|
  \sum _j \int_0 ^{c^* \epsilon} \sigma _{ij} (\bm{x} ^{\epsilon} _{s - \bm{c} \epsilon}) 
                       \Big( - ( y ^{\epsilon} _s )_j ds + d ( W_s )_j \Big) \Bigg| ^2 \left.\rule{0cm}{.7cm}\right] \\[3pt]
& \hspace{40pt} \leq C_1 \sum _j E \left[\rule{0cm}{.7cm}\right. \Bigg|
   \int_0 ^{c^* \epsilon} \sigma _{ij} (\bm{x} ^{\epsilon} _{s - \bm{c} \epsilon}) 
                        ( y ^{\epsilon} _s )_j ds \Bigg| ^2 \left.\rule{0cm}{.7cm}\right] \\
& \hspace{150pt} + C_1 \sum _j E \left[\rule{0cm}{.7cm}\right. \Bigg| \int_0 ^{c^* \epsilon} \sigma _{ij} (\bm{x} ^{\epsilon} _{s - \bm{c} \epsilon})  d ( W_s )_j \Bigg| ^2 \left.\rule{0cm}{.7cm}\right] .
\end{align*}
Using the Cauchy-Schwarz inequality, the It\^o isometry, and the boundedness of $\bm{\sigma}$ then gives
\begin{align*}
& E \left[\rule{0cm}{.8cm}\right. \left(\rule{0cm}{.7cm}\right. \sup_{0 \leq t \leq T} \Bigg|
  \left( \int_0 ^{c^* \epsilon} \bm{\sigma}(\bm{x} ^{\epsilon} _{s - \bm{c} \epsilon}) \epsilon \bm{D} ^{-1} d \bm{y} ^{\epsilon} _s \right) _i \Bigg| \left.\rule{0cm}{.7cm}\right) ^2 \left.\rule{0cm}{.8cm}\right] \hspace{140pt} \\[3pt]
& \hspace{20pt} \leq C_2 \epsilon \sum _j  
   \int_0 ^{c^* \epsilon} E \left[ \Big| \sigma _{ij} (\bm{x} ^{\epsilon} _{s - \bm{c} \epsilon}) 
                        ( y ^{\epsilon} _s )_j \Big| ^2 \right] ds  
  + C_1 \sum _j E \left[ \int_0 ^{c^* \epsilon} \sigma _{ij} ^2 (\bm{x} ^{\epsilon} _{s - \bm{c} \epsilon})  ds \right] \\[3pt]
& \hspace{20pt} \leq C_3 \epsilon \; .
\end{align*}

We now turn to the type $II$ terms, for which it suffices to show that
 for every $i, j, \ell,$ and $p$,
\begin{align}
&\lim_{\epsilon \rightarrow 0} E \left[\rule{0cm}{.8cm}\right. \left(\rule{0cm}{.7cm}\right. \sup_{0 \leq t \leq T} \Bigg| \int_0 ^t \Big( \sigma_{ij} (\bm{x} ^{\epsilon} _{s - \bm{c} \epsilon }) - \sigma_{ij} (\bm{x}_{s}^\epsilon) \Big)  d(W _s)_j \Bigg| \left.\rule{0cm}{.7cm}\right) ^2 \left.\rule{0cm}{.8cm}\right] = 0 \hspace{40pt} \label{typeIa} 
\end{align}
and
\begin{align}
& \lim_{\epsilon \rightarrow 0} E \left[\rule{0cm}{.8cm}\right. \sup_{c^* \epsilon \leq t \leq T} \left|\rule{0cm}{.7cm}\right. \int_{c^* \epsilon}^t \Bigg[ \frac{\partial  \sigma_{ij}}{\partial x_p} (\bm{x} ^{\epsilon} _{s - \bm{c} \epsilon }) \sigma_{p \ell} (\bm{x} ^{\epsilon} _{s - c_p \epsilon - \bm{c} \epsilon}) \notag \hspace{140pt} \\
&\hspace{100pt} - \frac{\partial \sigma_{ij}}{\partial x_p} (\bm{x} ^{\epsilon} _{s}) \sigma_{p \ell} (\bm{x} ^{\epsilon} _{s}) \Bigg] k _j \epsilon (y ^{\epsilon} _s)_j (y ^{\epsilon}_{s - c_p \epsilon })_{\ell}  ds \left.\rule{0cm}{.7cm}\right| \left.\rule{0cm}{.8cm}\right] = 0 \; . \label{typeIb}
\end{align}
We note that in (\ref{typeIb}) we take the supremum over $c^* \epsilon \leq t \leq T$ because in the preceding we have implicitly defined this term, as well as the type $III$ terms, to be zero for $0 \leq t < c^* \epsilon$.  First we show (\ref{typeIa}).  We will use the Lipschitz continuity of  $\sigma _{ij}$ for all $i, j$, which follows from the assumption of Theorem \ref{theorem} that $\bm{\sigma}$ has bounded first derivatives.  Also, we note that the It\^o integral in (\ref{typeIa}) is a martingale because the integrand is still nonanticipating in the presence of the delays.  Thus, by Doob's maximal inequality and the It\^o isometry, we have
\begin{align*}
E \left[\rule{0cm}{.8cm}\right. \left(\rule{0cm}{.7cm}\right. \sup_{0 \leq t \leq T} \Bigg| \int_0 ^t \Big( \sigma_{ij} (\bm{x} ^{\epsilon} _{s - \bm{c} \epsilon}) & - \sigma_{ij} (\bm{x}_{s}^\epsilon)  \Big)  d(W _s)_j \Bigg| \left.\rule{0cm}{.7cm}\right) ^2 \left.\rule{0cm}{.8cm}\right] \hspace{110pt} \\[3pt]
& \leq 2 E \left[\rule{0cm}{.7cm}\right. \Bigg( \int_0 ^T \Big( \sigma_{ij} (\bm{x} ^{\epsilon} _{s - \bm{c} \epsilon }) - \sigma_{ij} (\bm{x}_{s}^\epsilon) \Big)  d(W _s)_j \Bigg) ^2 \left.\rule{0cm}{.7cm}\right] \\[5pt]
& = 2 \int_0 ^T E \bigg[ \Big( \sigma_{ij} (\bm{x} ^{\epsilon} _{s - \bm{c} \epsilon }) - \sigma_{ij} (\bm{x}_{s}^\epsilon) \Big) ^2 \bigg] ds \; .
\end{align*}
We bound the above integral by using the Lipschitz continuity of $\sigma _{ij}$ and Lemma \ref{lemma:continuityofxepsilon}.  Since Lemma \ref{lemma:continuityofxepsilon} applies only to differences of values of the process $\bm{x} ^{\epsilon}$ at nonnegative times, we split the above integral into two terms.  The first term involves values of $\bm{x} ^{\epsilon}$ at negative times, i.e., values of the past condition $\bm{x} ^-$, and the second term only involves values of $\bm{x} ^{\epsilon}$ at nonnegative times.  Recalling that $c^* = \max_{1 \leq i \leq m} c_i$ and letting $M$ be the Lipschitz constant for $\sigma_{ij}$, we have
\begin{align*}
& E \left[\rule{0cm}{.8cm}\right. \left(\rule{0cm}{.7cm}\right.  \sup_{0 \leq t \leq T} \Bigg| \int_0 ^t \Big( \sigma_{ij} (\bm{x} ^{\epsilon} _{s - \bm{c} \epsilon}) - \sigma_{ij} (\bm{x}_{s}^\epsilon) \Big)  d(W _s)_j \Bigg| \left.\rule{0cm}{.7cm}\right) ^2 \left.\rule{0cm}{.8cm}\right] \hspace{150pt} \\[5pt]
& \hspace{10pt} \leq 2 \int_0 ^{c^* \epsilon} E \bigg[ \Big( \sigma_{ij} (\bm{x} ^{\epsilon} _{s - \bm{c} \epsilon }) - \sigma_{ij} (\bm{x}_{s}^\epsilon) \Big) ^2 \bigg] ds + 2 \int_{c^* \epsilon} ^T E \bigg[ \Big( M \left\| \bm{x} ^{\epsilon} _{s - \bm{c} \epsilon } - \bm{x}_{s}^\epsilon \right\| \Big) ^2 \bigg] ds \\[5pt]
& \hspace{10pt} \leq C_1 \epsilon + C_2 T \left( \epsilon + \epsilon ^{1 / 2} \right) 
\end{align*}
 using the boundedness of $\sigma _{ij}$ and Lemma \ref{lemma:continuityofxepsilon}, from which (\ref{typeIa}) follows.  To prove (\ref{typeIb}):
\begin{align*}
& E \left[\rule{0cm}{.8cm}\right. \sup_{c^* \epsilon \leq t \leq T} \left|\rule{0cm}{.7cm}\right. \int_{c^* \epsilon}^t \Bigg[ \frac{\partial \sigma_{ij}}{\partial x_p} (\bm{x} ^{\epsilon} _{s - \bm{c} \epsilon }) \sigma_{p \ell} (\bm{x} ^{\epsilon} _{s - c_p \epsilon - \bm{c} \epsilon }) \\
& \hspace{140pt} - \frac{\partial \sigma_{ij}}{\partial x_p} (\bm{x} ^{\epsilon} _{s}) \sigma_{p \ell} (\bm{x} ^{\epsilon} _{s}) \Bigg] k _j \epsilon (y ^{\epsilon} _s)_j (y ^{\epsilon}_{s - c_p \epsilon })_{\ell}  ds \left.\rule{0cm}{.7cm}\right| \left.\rule{0cm}{.8cm}\right] \\[5pt]
& \hspace{10pt} \leq  \int_{c^* \epsilon}^T E \left[\rule{0cm}{.8cm}\right. \left|\rule{0cm}{.7cm}\right. \Bigg[ \frac{\partial \sigma_{ij}}{\partial x_p} (\bm{x} ^{\epsilon} _{s - \bm{c} \epsilon }) \sigma_{p \ell} (\bm{x} ^{\epsilon} _{s - c_p \epsilon - \bm{c} \epsilon }) \\
& \hspace{150pt} - \frac{\partial \sigma_{ij}}{\partial x_p} (\bm{x} ^{\epsilon} _{s}) \sigma_{p \ell} (\bm{x} ^{\epsilon} _{s}) \Bigg] k _j \epsilon (y ^{\epsilon} _s)_j (y ^{\epsilon}_{s - c_p \epsilon })_{\ell}  \left.\rule{0cm}{.7cm}\right| \left.\rule{0cm}{.8cm}\right] ds  \\[5pt]
& \hspace{10pt} \leq \int_{c^* \epsilon}^T \left(\rule{0cm}{.7cm}\right. E \left[ \left( \frac{\partial \sigma_{ij}}{\partial x_p}  (\bm{x} ^{\epsilon} _{s - \bm{c} \epsilon }) \sigma_{p \ell} (\bm{x} ^{\epsilon} _{s - c_p \epsilon - \bm{c} \epsilon }) - \frac{\partial \sigma_{ij}}{\partial x_p} (\bm{x} ^{\epsilon} _{s}) \sigma_{p \ell} (\bm{x} ^{\epsilon} _{s}) \right)^2 \right] \left.\rule{0cm}{.7cm}\right) ^{1 / 2} \\[2pt]
& \hspace{180pt} \times \Bigg( E \left[ \left( k _j \epsilon (y ^{\epsilon} _s)_j (y ^{\epsilon}_{s - c_p \epsilon })_{\ell} \right)^2 \right] \Bigg) ^{1 / 2} ds \\[5pt]
& \hspace{10pt} \leq C_1 \int_{c^* \epsilon}^T \left(\rule{0cm}{.7cm}\right. E \left[ \left( \frac{\partial \sigma_{ij}}{\partial x_p} (\bm{x} ^{\epsilon} _{s - \bm{c} \epsilon }) \sigma_{p \ell} (\bm{x} ^{\epsilon} _{s - c_p \epsilon - \bm{c} \epsilon }) - \frac{\partial \sigma_{ij}}{\partial x_p} (\bm{x} ^{\epsilon} _{s}) \sigma_{p \ell} (\bm{x} ^{\epsilon} _{s}) \right)^2 \right] \left.\rule{0cm}{.7cm}\right) ^{1 / 2} ds
\end{align*}
by Lemma \ref{lemma:ey2toConstant}.  Again, in order to apply Lemma \ref{lemma:continuityofxepsilon}, we split the above integral into a part that involves values of the past condition $\bm{x} ^-$ and a part that only involves values of $\bm{x} ^{\epsilon}$ at nonnegative times.  We have
\begin{align*}
& E \left[\rule{0cm}{.8cm}\right. \sup_{c^* \epsilon \leq t \leq T} \left|\rule{0cm}{.7cm}\right. \int_{c^* \epsilon}^t \Bigg[ \frac{\partial \sigma_{ij}}{\partial x_p} (\bm{x} ^{\epsilon} _{s - \bm{c} \epsilon }) \sigma_{p \ell} (\bm{x} ^{\epsilon} _{s - c_p \epsilon - \bm{c} \epsilon }) \\
& \hspace{120pt} - \frac{\partial \sigma_{ij}}{\partial x_p} (\bm{x} ^{\epsilon} _{s}) \sigma_{p \ell} (\bm{x} ^{\epsilon} _{s}) \Bigg] k _j \epsilon (y ^{\epsilon} _s)_j (y ^{\epsilon}_{s - c_p \epsilon })_{\ell}  ds \left.\rule{0cm}{.7cm}\right| \left.\rule{0cm}{.8cm}\right] \hspace{7pt} \\[5pt]
& \hspace{2pt} \leq C_1 \int_{c^* \epsilon}^{2c^* \epsilon} \left(\rule{0cm}{.8cm}\right. E \left[\rule{0cm}{.7cm}\right. \Bigg( \frac{\partial \sigma_{ij}}{\partial x_p} (\bm{x} ^{\epsilon} _{s - \bm{c} \epsilon }) \sigma_{p \ell} (\bm{x} ^{\epsilon} _{s - c_p \epsilon - \bm{c} \epsilon }) \\
& \hspace{190pt} - \frac{\partial \sigma_{ij}}{\partial x_p} (\bm{x} ^{\epsilon} _{s}) \sigma_{p \ell} (\bm{x} ^{\epsilon} _{s}) \Bigg) ^2 \left.\rule{0cm}{.7cm}\right] \left.\rule{0cm}{.8cm}\right) ^{1 / 2} ds \\[5pt]
& \hspace{10pt} + C_1 \int_{2c^* \epsilon}^T \left(\rule{0cm}{.8cm}\right. E \left[\rule{0cm}{.7cm}\right. 2 \Bigg( \frac{\partial \sigma_{ij}}{\partial x_p} (\bm{x} ^{\epsilon} _{s - \bm{c} \epsilon }) \sigma_{p \ell} (\bm{x} ^{\epsilon} _{s - c_p \epsilon - \bm{c} \epsilon }) - \frac{\partial \sigma_{ij}}{\partial x_p} (\bm{x} ^{\epsilon} _{s - \bm{c} \epsilon}) \sigma_{p \ell} (\bm{x} ^{\epsilon} _{s}) \Bigg)^2 \\
& \hspace{90pt} + 2 \Bigg( \frac{\partial \sigma_{ij}}{\partial x_p} (\bm{x} ^{\epsilon} _{s - \bm{c} \epsilon }) \sigma_{p \ell} (\bm{x} ^{\epsilon} _{s}) - \frac{\partial \sigma_{ij}}{\partial x_p} (\bm{x} ^{\epsilon} _{s}) \sigma_{p \ell} (\bm{x} ^{\epsilon} _{s}) \Bigg) ^2 \left.\rule{0cm}{.7cm}\right] \left.\rule{0cm}{.8cm}\right) ^{1 / 2} ds \\[5pt]
& \hspace{2pt} \leq C_2 \epsilon + C_3 \int_{2c^* \epsilon}^T \left(\rule{0cm}{.7cm}\right. E \Bigg[ \big\| \bm{x} ^{\epsilon} _{s - c_p \epsilon - \bm{c} \epsilon } - \bm{x} ^{\epsilon} _{s} \big\| ^2 +  \left\| \bm{x} ^{\epsilon} _{s - \bm{c} \epsilon } - \bm{x} ^{\epsilon} _{s} \right\| ^2 \Bigg] \left.\rule{0cm}{.7cm}\right) ^{1 / 2} ds \\[2pt]
& \hspace{20pt} \mathrm{(using \; the \; boundedness \; and \; Lipschitz \; continuity \; of \;} \bm{\sigma} \mathrm{ \; and \; its \; first \; derivatives)} \\[2pt]
& \hspace{2pt} \leq C_2 \epsilon + C_4 T \left( \epsilon + \epsilon ^{1 / 2} \right) ^{1 / 2}
\end{align*}
by Lemma \ref{lemma:continuityofxepsilon}, which gives (\ref{typeIb}).  

For the type $III$ terms, it suffices to show that for every $i, j,$ and $p$,
\begin{equation}
\label{TypeIIA}
\lim_{\epsilon \rightarrow 0} E\left [ \sup_{0\leq t\leq T} \Big| \sigma_{ij} (\bm{x} ^{\epsilon} _{t - \bm{c} \epsilon }) k_j \epsilon (y ^{\epsilon} _t)_j \Big| ^2 \right] = 0 \hspace{5pt}
\end{equation}
\begin{equation}
\label{TypeIIB}
\lim_{\epsilon \rightarrow 0} E\left [ \sup_{0\leq t\leq T} 
          \Big| \sigma_{ij} (\bm{x} ^{\epsilon} _{c^* \epsilon - \bm{c} \epsilon }) k_j \epsilon (y_{c^* \epsilon} ^{\epsilon})_j \Big| ^2 \right] = 0 \hspace{5pt}
\end{equation}
and
\begin{equation}
\label{TypeIIC}
 \lim_{\epsilon \rightarrow 0} E\left [ \sup_{c^* \epsilon \leq t\leq T} \Bigg| \int_{c^* \epsilon} ^t \frac{\partial \sigma_{ij}}{\partial x_p} (\bm{x} ^{\epsilon} _{s - \bm{c} \epsilon }) k _j \epsilon (y ^{\epsilon} _s)_j f_p (\bm{x} ^{\epsilon} _{s - c_p \epsilon }) ds \Bigg|^2 \right] = 0 \; . \hspace{5pt}
\end{equation}
Equations (\ref{TypeIIA}) and (\ref{TypeIIB}) follow immediately from the boundedness of $\sigma_{ij}$ and Lemma \ref{lemma:stochasticconvolution}.  To prove equation (\ref{TypeIIC}), we first use the boundedness of $\frac{\partial \sigma_{ij}}{\partial x_p} $ and $f_p$ and then apply the Cauchy-Schwarz inequality:
\begin{align*}
& E \left[ \sup_{c^* \epsilon \leq t\leq T} \Bigg| \int_{c^* \epsilon} ^t \frac{\partial \sigma_{ij}}{\partial x_p} (\bm{x} ^{\epsilon} _{s - \bm{c} \epsilon }) k _j \epsilon (y ^{\epsilon} _s)_j f_p (\bm{x} ^{\epsilon} _{s - c_p \epsilon }) ds \Bigg|^2 \right] \\[2pt]
& \hspace{200pt} \leq C E\left [ \left( \int_{c^* \epsilon}^T \big| \epsilon (y ^{\epsilon} _s)_j \big| ds \right) ^2 \right] 
\\[2pt]
& \hspace{200pt} \leq C T \int_0^T E \left[ \big| \epsilon (y ^{\epsilon} _s)_j \big| ^2 \right] ds  
\end{align*}
which goes to zero by Lemma \ref{lemma:stochasticconvolution}.

Therefore $\bm{U} ^{\epsilon} \rightarrow \bm{0}$ as $\epsilon\rightarrow0$ in probability with respect 
to $C ([0,T], \mathbb{R} ^m)$, as claimed.

\subsubsection{$ \textbf{H} ^{\epsilon} \rightarrow \textbf{H}$}

Here we show that $\bm{H}^{\epsilon}\rightarrow \bm{H}$ in $L^2$, and thus in probability, with respect to $C([0,T],
 \mathbb{R} ^{(1 + n + 1 + m n^2)})$. Note that the first three components of $\bm{H}^\epsilon_t$, defined in (\ref{eq:DefH}), are independent of $\epsilon$. Thus, it suffices to show that $(G^{j\ell})_p \rightarrow0$ in $L^2$ with respect to $C([0,T] , \mathbb{R})$ for all $j,\ell = 1, ..., n$ and $p=1,...,m$.  
 
A heuristic argument that provides some insight into why $(G^{j\ell})_p$ converges to zero is as follows.  Recall that $(G^{j\ell} _t)_p$ is defined as
\begin{equation}
\big( G ^{j \ell} _t \big)_p = \int _0 ^t \bigg( k _j \epsilon ({y} ^{\epsilon} _s)_j ({y} ^{\epsilon}_{s - c_p \epsilon })_{\ell}  -E \Big[ k_j\epsilon(y_s^\epsilon)_j ({y} ^{\epsilon}_{s - c_p \epsilon })_{\ell} \Big] \bigg) ds \; . 
\end{equation}
For each $j$, define the new process $(z)_j$ by $(z_r)_j = \sqrt{\epsilon} (y^{\epsilon}_{\epsilon r})_j$.  Then $(z)_j$ solves the
$\epsilon$-independent equation
$$d(z_r)_j = - \frac{1}{k_j} (z_r)_j dr + \frac{1}{k_j} d \big( \tilde{W} _ r \big) _j $$
with the Wiener process $\tilde{W}_r = \epsilon^{- 1/2}W_{\epsilon r}$.
In terms of the process $(z)_j$, $(G^{j\ell} _t )_p$ can be written as
$$ \big( G^{j\ell} _t \big) _p = \epsilon \int _0 ^{t / \epsilon} \bigg( k _j ({z} _u)_j ({z} _{u - c_p })_{\ell}  - E \Big[ k_j (z_u)_j ({z} _{u - c_p })_{\ell} \Big] \bigg)du \; .$$
The above integral can be thought of as the sum of  $ m = O\left(1/ \epsilon\right)$ identically distributed random variables with zero mean.  Furthermore, these random variables are weakly correlated, since the covariance function for $(z)_j$ decays exponentially with an exponential decay constant of order 1.  Thus, we expect the $L^2$-norm of this sum to grow about as fast as $O\left(1/ \sqrt{\epsilon} \right)$.  Since  $(G^{j\ell} _t)_p$ is equal to this integral multiplied by $\epsilon$, we expect $(G^{j\ell} _t )_p$ to converge to zero as $\epsilon \rightarrow 0$ for all $t \in [0,T]$.  For fixed $t \in [0,T]$, convergence of $(G^{j\ell} _t )_p$ to zero in $L^2$ can be shown by expressing the square of the integral as a double integral and then using Wick's theorem to compute the expectation of the terms in the integrand. However, to control the supremum norm, more work is required, and this is the content of the next lemma.

 \begin{lemma}
\label{lemma:uniformaverage}
For each $\epsilon > 0$, let $\bm{y}^\epsilon$ be defined on $t \geq 0$ as the solution to equation~(\ref{yequation}) with $\bm{y}^\epsilon _0$ distributed according to the stationary distribution corresponding to (\ref{yequation}), and let $\bm{y}^\epsilon$ be defined on $t_- \leq t < 0$ in such a way so that $\bm{y}^\epsilon$ considered on the interval $t \geq t_-$ is a stationary process.  Then for all $j, \ell ,$ and $p$,
\begin{equation}
\label{uniformaveragestatement1}
\lim _{\epsilon \rightarrow 0} E \left[ \Bigg( \sup_{0 \leq t \leq T} \Bigg| \int_0^t  \left( k_j \epsilon (y^{\epsilon}_u)_j (y^{\epsilon}_{u - c_p \epsilon})_j - \frac{1}{2} e^{- \frac{c_p}{k_j}} \right) du \Bigg| \Bigg) ^2 \right] = 0
\end{equation}
and
\begin{equation}
\label{uniformaveragestatement2}
\lim _{\epsilon \rightarrow 0} E \left[ \Bigg( \sup_{0 \leq t \leq T} \Bigg| \int_0^t   k_j \epsilon (y^{\epsilon}_u)_j (y^{\epsilon}_{u - c_p \epsilon})_{\ell} \;du \Bigg| \Bigg) ^2 \right] = 0 \quad \mbox{for }j\neq \ell \; .
\end{equation}
\end{lemma}

\vspace{10pt}

To show this, a Riemann sum approximation is used in order
to apply the following maximal inequality for sums of a stationary sequence of random variables. 

\begin{lemma}\cite[Proposition 2.3]{peligrad2005}
\label{lemma:Peligrad}
Let $\{X_i \; : \; i \in \mathbb{N} \}$ be a stationary sequence of random variables and let $S_n = \sum _{i = 1}^n X_i$.  Let $\mathcal{F}_0$ be the $\sigma$-algebra generated by $X_1$. Let $n \in \mathbb{N}$ and let $r = \lceil \log _2 (n) \rceil$ (i.e., $r$ is the smallest integer greater than or equal to $\log _2 (n)$).  Then
\begin{equation}
\label{maximal inequality1}
E \left[ \max_{1 \leq i \leq n} S_i ^2 \right] \leq n \left( 2 \sqrt{E \big[ X_1 ^2 \big]} + \left( 1 + \sqrt{2} \right) \Delta _r \right) ^2 
\end{equation}
where
$$\Delta _r = \sum _{j = 0} ^{r - 1} \frac{\sqrt{E\left [E \big[ S_{2^j} | \mathcal{F}_0\big]^2\right ]}}{2^{j / 2}} \; \; . $$
\end{lemma}

\begin{remark}
By the conditional Jensen's inequality, $ E \big[ S_{2^j} | \mathcal{F}_0\big]^2  \leq E \big[ S_{2^j} ^2 | \mathcal{F}_0\big]$.
Thus, (\ref{maximal inequality1}) implies
\begin{equation}
\label{maximal inequality2}
E \left[ \max_{1 \leq i \leq n} S_i ^2 \right] \leq n \left( 2 \sqrt{E \big[ X_1 ^2 \big]} + \left( 1 + \sqrt{2} \right) 
      \tilde{\Delta} _r \right) ^2 
\end{equation}
where
\begin{equation*}
\tilde{\Delta} _r 
   = \sum _{j = 0} ^{r - 1} \frac{\sqrt{E\big [ S_{2^j} ^2 \big ]}}{2^{j / 2}} \; \; .
\end{equation*}
\end{remark}

\begin{proof}[Proof of Lemma~\ref{lemma:uniformaverage}]
Let 
$$\Psi _{j \ell p} = E \left[k_j \epsilon (y^{\epsilon}_u)_j (y^{\epsilon}_{u - c _p \epsilon})_{\ell} \right] =  \frac{\delta _{j \ell}}{2} e^{- \frac{c_p}{k_j}} \; .$$
Then showing both (\ref{uniformaveragestatement1}) and (\ref{uniformaveragestatement2}) is equivalent to showing that for all $j, \ell ,$ and $p$,
\begin{equation}
\label{uniformaveragecombinedstatement}
\lim _{\epsilon \rightarrow 0} E \left[ \left(\rule{0cm}{.7cm}\right. \sup_{0 \leq t \leq T} \Bigg| \int_0^t  \left( k_j \epsilon (y^{\epsilon}_u)_j (y^{\epsilon}_{u - c _p \epsilon})_{\ell} - \Psi _{j \ell p} \right) du \Bigg| \left.\rule{0cm}{.7cm}\right) ^2 \right] = 0 \; .
\end{equation}
Define the new process $\bm{z} = \big( (z)_1, ... , (z)_n \big)$ by defining each component 
$(z)_j$ by $(z_r)_j = \sqrt{\epsilon} (y^{\epsilon}_{\epsilon r})_j$.  Then $(z)_j$ solves the
$\epsilon$-independent equation
$$d(z_r)_j = - \frac{1}{k_j} (z_r)_j dr + \frac{1}{k_j} d \big( \tilde{W} _ r \big)_j \; , \quad (z_0)_j = \sqrt{\epsilon}(y_0^\epsilon)_j$$
with the Wiener process $\tilde{W}_r = \epsilon^{-1/2}W_{\epsilon r}$.
Letting $r = u / \epsilon$ and $N = 1 / \epsilon$ gives
\begin{align*} 
\int_0^t  \left( k_j \epsilon (y^{\epsilon}_u)_j (y^{\epsilon}_{u - c_p \epsilon})_{\ell} - \Psi _{j \ell p} \right) du & = \epsilon \int_0^{t / \epsilon}  \left( k_j \epsilon (y^{\epsilon}_{\epsilon r})_j (y^{\epsilon}_{\epsilon r - c_p \epsilon})_{\ell} - \Psi _{j \ell p} \right) dr \\[3pt]
& = \frac{1}{N} \int_0^{Nt}  \Big( k_j (z_r)_j (z_{r - c_p})_{\ell} - \Psi _{j \ell p} \Big) dr \; .
\end{align*}
We approximate the above integral by a Riemann sum.  Let $\{t_i\}_{i=1,...,m}$ be the partition of $[0,NT]$ into $m$ equal parts of size
$\Delta t = t_{i+1}-t_{i}$, where $t_1 = 0$ and $t_{i+1}-t_{i}= NTm^{-1}$ for all $i$. Near the end of the proof, we will choose $m$ in terms of $N = 1 / \epsilon$.  Let $b(t) = \max \{i \; : \; iNTm^{-1} \leq Nt \}$.  We add and subtract the corresponding Riemann sum and use the Cauchy-Schwarz inequality to obtain
\begin{align}
E \left[\rule{0cm}{.7cm}\right. & \sup_{0\leq t\leq T} \Bigg| \int_0^t \Big( k_j \epsilon(y_u^\epsilon)_j(y^{\epsilon}_{u-c_p\epsilon})_{\ell} - \Psi _{j \ell p} \Big) du \Bigg| ^2 \left.\rule{0cm}{.7cm}\right] \nonumber \hspace{170pt} \\[10pt]
\label{Lemma6eq1}
 & \leq 2 E \left[\rule{0cm}{.8cm}\right. \sup_{0\leq t\leq T} \left|\rule{0cm}{.7cm}\right.  \frac{1}{N} \left(\rule{0cm}{.7cm}\right. \int_0^{Nt} \Big( k_j (z_r)_j (z_{r-c_p})_{\ell} - \Psi _{j \ell p} \Big) dr  \nonumber \\
& \hspace{110pt} - \left. \left. \left. \frac{NT}{m}\sum_{i=1}^{b(t)} \Big( k_j(z_{t_i})_j(z_{t_{i}-c_p})_{\ell} - \Psi _{j \ell p} \Big) \right) \right |^2\right] \nonumber \\
& \hspace{10pt} + \frac{2}{N ^2} E\left [\sup_{0\leq t\leq T}\left |  \frac{NT}{m}\sum_{i=1}^{b(t)}\Big( k_j(z_{t_i})_j(z_{t_{i}-c_p})_{\ell} - \Psi _{j \ell p} \Big) \right| ^2 \right] \; .
\end{align}
Each of the two terms on the right-hand side converges to zero for a different reason. The first term goes to zero because as $m$ increases, the 
Riemann sum better approximates the integral.  The second term goes to zero because the sum grows like $\sqrt{m}$ since it is a sum of on the order of $m$ weakly correlated, mean-zero random variables.

We will start with the first term. First, writing
$$ \frac{NT}{m} \Big( k_j(z_{t_i})_j (z_{t_i-c_p})_{\ell} - \Psi _{j \ell p}  \Big) =  \int_{t_i}^{t_{i+1}}\Big( k_j(z_{t_i})_j(z_{t_i-c_p})_{\ell} - \Psi _{j \ell p} \Big) \; dr$$
we have
\begin{align}
 E & \left[\rule{0cm}{.8cm}\right. \sup_{0\leq t\leq T} \left|\rule{0cm}{.7cm}\right.  \frac{1}{N} \int_0^{Nt} 
         \Big( k_j(z_r)_j(z_{r-c_p})_{\ell} - \Psi _{j \ell p} \Big) \;dr \notag \hspace{150pt} \\
   & \hspace{125pt} - \frac{T}{m}\sum_{i=1}^{b(t)} \Big( k_j(z_{t_i})_j (z_{t_{i}-c_p})_{\ell} -\Psi _{j \ell p} \Big) \left.\rule{0cm}{.7cm}\right| ^2 \left.\rule{0cm}{.8cm}\right] \notag \hspace{15pt} \\[10pt]
 & = E \left[\rule{0cm}{.8cm}\right. \sup_{0\leq t\leq T} \left|\rule{0cm}{.7cm}\right. \frac{1}{N}\sum_{i=1}^{b(t)}\int_{t_i}^{t_{i+1}} \bigg[ 
       \Big( k_j(z_r)_j(z_{r-c_p})_{\ell} - \Psi _{j \ell p} \Big) \notag \\
 & \hspace{135pt} -  \Big( k_j(z_{t_i})_j(z_{t_{i}-c_p})_{\ell} -\Psi _{j \ell p} \Big) \bigg] \;dr \notag \\
 & \hspace{125pt} + \frac{1}{N} \int_{t_{b(t)}}^{Nt} \Big( k_j(z_r)_j(z_{r-c_p})_{\ell} - \Psi _{j \ell p} \Big) dr \left|\rule{0cm}{.7cm}\right. ^2 \left.\rule{0cm}{.8cm}\right] \notag \\[5pt]
& \leq  2 E \left[\rule{0cm}{.8cm}\right. \sup_{0\leq t\leq T} \left|\rule{0cm}{.7cm}\right. \frac{1}{N}\sum_{i=1}^{b(t)}\int_{t_i}^{t_{i+1}} \bigg[ 
   \Big( k_j(z_r)_j(z_{r-c_p})_{\ell} - \Psi _{j \ell p} \Big) \notag \\
 & \hspace{150pt} -  \Big( k_j(z_{t_i})_j(z_{t_{i}-c_p})_{\ell} - \Psi _{j \ell p} \Big) \bigg] \;dr \left|\rule{0cm}{.7cm}\right. ^2 \left.\rule{0cm}{.8cm}\right] \notag \\
 & \hspace{10pt} + 2 E \left[\rule{0cm}{.8cm}\right. \sup_{0\leq t\leq T} \left|\rule{0cm}{.7cm}\right. \frac{1}{N} \int_{t_{b(t)}}^{Nt} \Big( k_j(z_r)_j(z_{r-c_p})_{\ell} - \Psi _{j \ell p} \Big) dr \left|\rule{0cm}{.7cm}\right. ^2 \left.\rule{0cm}{.8cm}\right] .  \label{uniformterms1and2}
 \end{align}
We first bound the first term on the right-hand side of (\ref{uniformterms1and2}).  By the Cauchy--Schwarz inequality,
\begin{align*}
E \left[\rule{0cm}{.8cm}\right. & \sup_{0\leq t\leq T} \left|\rule{0cm}{.7cm}\right. \frac{1}{N}\sum_{i=1}^{b(t)}\int_{t_i}^{t_{i+1}} \bigg[ \Big (k_j(z_r)_j(z_{r-c_p})_{\ell} - \Psi _{j \ell p} \Big) \\
& \hspace{150pt} -  \Big( k_j (z_{t_i})_j (z_{t_{i}-c_p})_{\ell} - \Psi _{j \ell p} \Big) \bigg] \;dr \left|\rule{0cm}{.7cm}\right. ^2 \left.\rule{0cm}{.8cm}\right] \\[5pt]
& \leq E \left[\rule{0cm}{.8cm}\right. \left(\rule{0cm}{.7cm}\right. \frac{1}{N}\sum_{i=1}^{b(T)}\int_{t_i}^{t_{i+1}} \bigg| \Big( k_j(z_r)_j(z_{r-c_p})_{\ell} - \Psi _{j \ell p} \Big) \\
& \hspace{150pt} -  \Big( k_j(z_{t_i})_j(z_{t_{i}-c_p})_{\ell} - \Psi _{j \ell p} \Big) \bigg| \;dr \left.\rule{0cm}{.7cm}\right) ^2 \left.\rule{0cm}{.8cm}\right] \\[3pt]
& \leq  \frac{1}{N ^2} b(T) \Delta t \sum_{i=1}^{b(T)}\int_{t_i}^{t_{i+1}} E \left[\rule{0cm}{.7cm}\right. \Bigg( \Big( k_j(z_r)_j(z_{r-c_p})_{\ell} - \Psi _{j \ell p} \Big) \\
& \hspace{160pt} -  \Big( k_j(z_{t_i})_j(z_{t_{i}-c_p})_{\ell} - \Psi _{j \ell p} \Big) \Bigg) ^2 \left.\rule{0cm}{.7cm}\right] \;dr \; .
\end{align*}
Note that $b(T) = m$.
 We compute the expectation that makes up the integrand. First we expand the square:
\begin{align*}
	E & \left[ \Bigg( \Big( k_j(z_r)_j(z_{r-c_p})_{\ell} - \Psi _{j \ell p} \Big) - \Big( k_j(z_{t_i})_j(z_{t_{i}-c_p})_{\ell} - \Psi _{j \ell p} \Big) \Bigg) ^2 \right] \hspace{80pt} \\[3pt]
	& \hspace{40pt} = E\left[ \Big(	k_j(z_r)_j (z_{r-c_p})_{\ell} - \Psi _{j \ell p} \Big) ^2 \right ]
	+ E \left[ \Big( k_j (z_{t_i})_j (z_{t_{i}-c_p})_{\ell} - \Psi _{j \ell p} \Big)^2\right ]\\[3pt]
	& \hspace{50pt} - 2E\bigg[ \Big( k_j(z_r)_j (z_{r-c_p})_{\ell} - \Psi _{j \ell p} \Big) \Big( k_j (z_{t_i})_j (z_{t_{i}-c_p})_{\ell} - \Psi _{j \ell p} \Big)\bigg] \; . 
\end{align*}
By using Wick's Theorem \cite{janson} to compute each expectation, we obtain
\begin{align*}
E & \left [ \Bigg( \Big( k_j(z_r)_j (z_{r-c_p})_{\ell} - \Psi _{j \ell p} \Big) - \Big( k_j (z_{t_i})_j (z_{t_{i}-c_p})_{\ell} - \Psi _{j \ell p} \Big) \Bigg) ^2 \right] \hspace{65pt} \\[3pt]
& \hspace{25pt} = \frac{k_j}{2 k_{\ell}} \left( 1-e^{-\frac{|r-t_i|}{k_j}} e^{-\frac{|r-t_i|}{k_{\ell}}} \right) + \frac{\delta _{j \ell}}{2} \left( e^{-2\frac{c_p}{k_j}} -e^{-\frac{|r-t_i + c_p|}{k_j}} e^{-\frac{|r-t_i - c_p|}{k_j}} \right) . 
\end{align*}
Note that $r\geq t_i$ and $r-t_i \leq \Delta t$, so that for $\Delta t < \min_p c_p$  ,
$$ - \frac{|r-t_i + c_p|}{k_j} - \frac{ |r-t_i - c_p|}{k_j} = - \frac{r - t_i + c_p}{k_j} + \frac{r - t_i - c_p}{k_j} = - 2 \frac{c_p}{k_j} \; . $$
Therefore, letting $C_1 = \max_{j, \ell} \frac{k_j}{2 k_{\ell}}$ and $C_2 = \max_{j, \ell} \frac{k_j + k_{\ell}}{k_j k_{\ell}}$, we have, for $\Delta t = \frac{NT}{m}$ sufficiently small,
\begin{align*}
& E\left [ \Bigg( \Big( k_j(z_r)_j (z_{r-c_p})_{\ell} - \Psi _{j \ell p} \Big) - \Big( k_j(z_{t_i})_j(z_{t_{i}-c_p})_{\ell} - \Psi _{j \ell p} \Big) \Bigg) ^2 \right ] \hspace{60pt} \\
 & \hspace{245pt} \leq C_1 \left( 1-e^{- C_2 \frac{NT}{m}} \right) . 
\end{align*}
Thus, for $N / m$ sufficiently small, we have the following bound for the first term in (\ref{uniformterms1and2}):
\begin{align}
2E & \left[\rule{0cm}{.8cm}\right. \sup_{0\leq t\leq T} \left|\rule{0cm}{.7cm}\right. \frac{1}{N} \sum_{i=1}^{b(t)}\int_{t_i}^{t_{i+1}} \bigg[ \Big (k_j(z_r)_j(z_{r-c_p})_{\ell} - \Psi _{j \ell p} \Big) \notag \hspace{150pt} \\
& \hspace{140pt} -  \Big( k_j (z_{t_i})_j (z_{t_{i}-c_p})_{\ell} - \Psi _{j \ell p} \Big) \bigg] \;dr \left|\rule{0cm}{.7cm}\right. ^2 \left.\rule{0cm}{.8cm}\right] \notag \\
& \leq \frac{2}{N ^2} m \Delta t \sum_{i=1}^{m}\int_{t_i}^{t_{i+1}}  E \left[\rule{0cm}{.7cm}\right. \Bigg( \Big( k_j(z_r)_j(z_{r-c_p})_{\ell} - \Psi _{j \ell p} \Big) \notag \\
& \hspace{140pt} -  \Big( k_j(z_{t_i})_j(z_{t_{i}-c_p})_{\ell} - \Psi _{j \ell p} \Big) \Bigg) ^2 \left.\rule{0cm}{.7cm}\right] \;dr \notag \\
&  \leq \frac{2}{N ^2} m ^2 (\Delta t)^2 C_1 \left( 1-e^{- C_2 \frac{NT}{m}} \right) \notag \\[5pt]
\label{uniformfirsttermbound}
& = 2C_1 T^2 \left( 1-e^{- C_2 \frac{NT}{m}} \right)  .
 \end{align}

We now bound the second term in (\ref{uniformterms1and2}):
\begin{align*}
2 E \left[\rule{0cm}{.7cm}\right. \sup_{0\leq t\leq T} \Bigg| \frac{1}{N} & \int_{t_{b(t)}}^{Nt} \Big( k_j(z_r)_j(z_{r-c_p})_{\ell} - \Psi _{j \ell p} \Big) dr \Bigg| ^2 \left.\rule{0cm}{.7cm}\right] \hspace{120pt} \\[3pt]
& \hspace{5pt}\leq \frac{2}{N^2} E \left[ \max_{1 \leq i \leq m} \Bigg( \int_{t_i}^{t_{i + 1}} \Big| k_j (z_r)_j (z_{r-c_p})_{\ell} - \Psi _{j \ell p} \Big| dr \Bigg) ^2 \right]
\end{align*}
\begin{align}
 \hspace{30pt} & \leq \frac{2}{N^2} E \left[ \sum_{i = 1}^m \Bigg( \int_{t_i}^{t_{i + 1}} \Big| k_j (z_r)_j (z_{r-c_p})_{\ell} - \Psi _{j \ell p} \Big| dr \Bigg) ^2 \right] \notag \\
& \leq \frac{2}{N^2}C_3 m(\Delta t)^2 \notag \\[5pt]
\label{uniformsecondtermbound}
& = \frac{2 C_3 T^2}{m} \; .
\end{align}

Together, (\ref{uniformfirsttermbound}) and (\ref{uniformsecondtermbound}) give a bound for the first term in (\ref{Lemma6eq1}). We now turn our attention to the second term in (\ref{Lemma6eq1}).  To bound this term we will use Lemma \ref{lemma:Peligrad}.  First, we note
\begin{align*}
\frac{1}{N ^2} E & \left [\sup_{0\leq t\leq T}\left |  \frac{NT}{m}\sum_{i=1}^{b(t)}\Big( k_j(z_{t_i})_j (z_{t_{i}-c_p})_{\ell} - \Psi _{j \ell p} \Big) \right| ^2 \right] \hspace{125pt} \\[5pt]
 & \hspace{90pt} =  E\left [\max_{0\leq n \leq m}\left |  \frac{T}{m} \sum_{i=1}^{n} \Big( k_j(z_{t_i})_j (z_{t_{i}-c_p})_{\ell} - \Psi _{j \ell p} \Big) \right| ^2 \right].
\end{align*}
Using the notation of Lemma \ref{lemma:Peligrad}, let $X_i = k_j(z_{t_i})_j (z_{t_{i}-c_p})_{\ell} - \Psi _{j \ell p}$.  Note 
that $\{X_i \}$ is a stationary sequence since $\bm{z}$ is a stationary process.
We now estimate the quantity $\tilde{\Delta} _r$ in (\ref{maximal inequality2}).  Consider the Riemann sum
$$ \frac{NT}{m} \sum _{i = 1} ^{m}  \Big( k_j (z_{t_i})_j (z_{t_i - c_p})_{\ell} - \Psi _{j \ell p} \Big) \; . $$
We have
\begin{equation*}
E \left[ \left( \frac{NT}{m}  \sum_{i = 1}^{m}  \Big( k_j (z_{t_i})_j (z_{t_i - c_p})_{\ell} - \Psi _{j \ell p} \Big)  \right) ^2 \right] \hspace{180pt} 
\end{equation*}
$$ =  E \left[\rule{0cm}{.7cm}\right. \sum_{i = 1}^{m} \frac{NT}{m}  \sum_{q = 1}^{m} \frac{NT}{m} \Big( k_j (z_{t_i})_j (z_{t_i - c_p})_{\ell} - \Psi _{j \ell p} \Big) \Big( k_j (z_{t_q})_j (z_{t_q - c_p})_{\ell} - \Psi _{j \ell p} \Big) \left.\rule{0cm}{.7cm}\right] $$
$$ = \sum_{i = 1}^{m} \frac{NT}{m}  \sum_{q = 1}^{m} \frac{NT}{m} E \Bigg[ k_j^2 (z_{t_i})_j (z_{t_i - c_p})_{\ell} (z_{t_q})_j (z_{t_q - c_p})_{\ell} -  \Psi _{j \ell p} k_j (z_{t_i})_j (z_{t_i - c_p})_{\ell} \hspace{9pt} $$
$$ \hspace{193pt} - \hspace{2pt} \Psi _{j \ell p} k_j (z_{t_q})_j (z_{t_q - c_p})_{\ell} + \left( \Psi _{j \ell p} \right) ^2 \Bigg] $$
$$ = \sum_{i = 1}^{m} \frac{NT}{m}  \sum_{q = 1}^{m} \frac{NT}{m} \Bigg[ \frac{k_j}{4k_{\ell}} e^{-  \frac{|t_i - t_q|}{k_j}} e^{- \frac{|t_i - t_q|}{k_{\ell}}} +  \frac{\delta_{j \ell}}{4}  e^{- \frac{|t_i - t_q + c_p|}{k_j}} e^{- \frac{|t_i - t_q - c_p|}{k_j}} \Bigg] \hspace{30pt} $$
$$\leq C_4 NT \hspace{300pt}$$ 
where $C_4$ does not depend on $N, m,$ or $T$, by comparison with the corresponding integral (recall that ${NT \over m} = t_{i+1} - t_i$).
Thus,
$$E \left[\rule{0cm}{.7cm}\right. \left( \sum_{i = 1}^{m}  \Big( k_j (z_{t_i})_j (z_{t_i - c_p})_{\ell} - \Psi _{j \ell p} \Big)  \right) ^2 \left.\rule{0cm}{.7cm}\right] \leq C_4 \frac{m^2}{NT} \; .$$
Applying the last bound to partial sums $S_{2^j} = \sum_{i=1}^{2^j} X_i$, with $NT$ replaced by ${2^j \over m}NT$, we obtain
$$
\tilde{\Delta} _r 
   = \sum _{j = 0} ^{r - 1} \frac{\sqrt{E\big [ S_{2^j} ^2 \big ]}}{2^{j / 2}} \leq C_5 \sum_{j=0}^{r-1}{2^j \over \sqrt{{2^j \over m}NT}}2^{-j/2} = C_5 r \sqrt{m \over NT} \; .
$$
Thus, by Lemma \ref{lemma:Peligrad}
\begin{align*}
& E \left[ \max_{1 \leq n \leq m} \left( \sum_{i = 1}^{n}  \Big( k_j (z_{t_i})_j (z_{t_i - c_p})_{\ell} - \Psi _{j \ell p} \Big)  \right)^2 \right] \hspace{160pt} \\[3pt]
& \hspace{145pt} \leq m \Bigg( 2 C_6 + \left( 1 + \sqrt{2} \right) C_5 \sqrt{\frac{m}{NT}} \log  _2 (m) \Bigg)^2 
\end{align*}
where $C_6 = \left( E\left[ \Big( k_j (z_{t_i})_j (z_{t_i - c_p})_{\ell} - \Psi _{j \ell p} \Big)^2 \right] \right) ^{1 / 2}$.  From this we obtain the following, which provides a bound for the second term in (\ref{Lemma6eq1}):
\begin{align}
&E \left[ \max_{1 \leq n \leq m} \left( \frac{T}{m} \sum_{i = 1}^{n}  \Big( k_j (z_{t_i})_j (z_{t_i - c_p})_{\ell} - \Psi _{j \ell p} \Big)  \right)^2 \right] \hspace{155pt} \notag \\[3pt]
\label{uniformthirdtermbound}
& \hspace{105pt} \leq \frac{T^2}{m} \Bigg( 2 C_6 + \left( 1 + \sqrt{2} \right) C_5 \sqrt{\frac{m}{NT}} \log _2 (m) \Bigg)^2.
\end{align}

Putting together (\ref{Lemma6eq1}), (\ref{uniformterms1and2}), (\ref{uniformfirsttermbound}), (\ref{uniformsecondtermbound}), and (\ref{uniformthirdtermbound}), we have, for $N / m$ sufficiently small, 
\begin{align*}
E & \left [\sup_{0\leq t\leq T} \Bigg| \int_0^t \Big( k_j \epsilon(y_u^\epsilon)_j (y^{\epsilon}_{u-c_p\epsilon})_{\ell} - 
\Psi _{j \ell p} \Big) \;du \Bigg| ^2 \right] \hspace{170pt} \\[3pt]
 & \hspace{20pt} \leq 2 C_1 T^2 \left( 1-e^{- C_2 \frac{NT}{m}} \right) + \frac{2 C_3 T^2}{m} \\
 & \hspace{120pt} + \frac{2 T^2}{m} \Bigg( 2 C_6 + \left( 1 + \sqrt{2} \right) C_5 \sqrt{\frac{m}{NT}} \log _2 (m) \Bigg)^2 \; .
\end{align*}
Setting $m = N^2$ proves (\ref{uniformaveragecombinedstatement}) and thus (\ref{uniformaveragestatement1}) and (\ref{uniformaveragestatement2}).
\end{proof}

\subsubsection{Conditions \ref{condition} and \ref{condition2}}

Now we show that Conditions \ref{condition} and \ref{condition2} are satisfied. First we show that $\bm{H}^\epsilon$ satisfies Condition \ref{condition}. To do so, we must show that for every $j, p,$ and $\ell \neq j$,
\begin{equation*}
  \int _0 ^T \bigg| k _j \epsilon (y ^{\epsilon} _s)_j (y ^{\epsilon}_{s - c_p \epsilon })_j  - \frac{1}{2} e^{- \frac{c_p}{k_j}}  \bigg| ds 
\end{equation*}
and 
\begin{equation*}
  \int _0 ^T \Big| k _j \epsilon (y ^{\epsilon} _s)_j (y ^{\epsilon}_{s - c_p \epsilon })_{\ell} \Big| ds  
\end{equation*}
are stochastically bounded. This follows from Lemma~\ref{lemma:ey2toConstant} by the Chebyshev inequality, since the second moments of the integrands are bounded by a constant, independent of $\epsilon$. 

Condition \ref{condition2} is the requirement that $\bm{h}$ is continuous.  In view of the definition (\ref{hdefinition}) of $\bm{h}$, this is true by the assumptions of Theorem \ref{theorem} that $\bm{f}$ is continuous and that $\bm{\sigma}$ has continuous first derivatives.

We now complete the proof of Theorem \ref{theorem}.  From Section~\ref{sec:Utozero}, $ \bm{U} ^\epsilon\rightarrow \bm{0}$ in probability with respect to $C([0,T], \mathbb{R}^m)$. 
From Lemma~\ref{lemma:uniformaverage}, $\bm{H}^\epsilon\rightarrow
\bm{H}$ in $L^2$ with respect to $C([0,T], \mathbb{R}^{(1 + n + 1 + mn^2)})$. 
 Therefore, $(\bm{U}^\epsilon,\bm{H}^\epsilon)\rightarrow (\bm{0},\bm{H})$ in probability with respect to $C([0,T],\mathbb{R}^m \times \mathbb{R}^{(1 + n + 1 + mn^2)})$. Furthermore $\bm{H}^\epsilon$ satisfies Condition \ref{condition}
 by Lemma~\ref{lemma:ey2toConstant} and $\bm{h}$ satisfies Condition \ref{condition2}. Thus, by Lemma~\ref{theorem:KP}, $\bm{x}^\epsilon
 \rightarrow\bm{x}$ in probability with respect to $C([0,T],\mathbb{R}^m)$. 
\end{proof}
 
\section{Discussion}\label{sec:Discussion}
 
We have derived the limiting equation for a general stochastic differential delay equation with multiplicative colored noise as the time delays and correlation times of the noises go to zero at the same rate.  As a result of the dependence of the noise coefficients on the state of the system (multiplicative noise), a \emph{noise-induced drift} appears in the limiting equation.  This result is useful for applications as the limiting SDE could provide a model that is easier to analyze than the original equation and at the same time still accounts for the effects of the time delays through the coefficients of the noise-induced drift.

The noise-induced drift has a form analogous to that of the \emph{Stratonovich correction} to the It\^o equation with noise term $ \bm{\sigma} (\bm{x}_t) d \bm{W} _t$.  That is, the noise-induced drift is a linear combination of the terms $\sigma_{pj} (\bm{x} _t) \frac{\partial \sigma _{ij} }{\partial x_p} (\bm{x} _t)$ . The coefficients of this linear combination in the limiting equation~\eqref{limiting equation} are 
\begin{equation}\label{alpha}
\alpha _{jp} = {1\over2} e^{- \frac{\delta _p}{\tau_j}},
\end{equation}
whereas the coefficients of the Stratonovich correction would all be equal to $1 / 2$.
Thus, as explained in \cite{mcdaniel2016,pesce2013}, one can interpret the terms of the noise-induced drift as representing different stochastic integration conventions.  We note that the exponential factor in (\ref{alpha}) comes from the form (i.e., exponentially decaying) of the covariance function of the particular noise process that we use, that is, the Ornstein-Uhlenbeck process.

The limiting equation that we have derived here is a more accurate approximation of the delayed system than the limiting equation derived in \cite{mcdaniel2016,pesce2013}. In particular, in \cite{mcdaniel2016,pesce2013}, equation \eqref{SDDE} was Taylor expanded to first-order in the time delay and the limiting equation corresponding to the expanded equation was derived.  It was shown that this limiting equation contains a noise-induced drift $\bm{\tilde{S}}(\bm{x}_t)$ defined componentwise as
\begin{equation}
	\tilde{S} _i(\bm{x}) = \sum _{p,j} \frac{1}{2} \left( 1 + \frac{\delta _p}{\tau _j} \right)^{-1}  \sigma_{pj} (\bm{x}) \frac{\partial \sigma _{ij} }{\partial x_p} (\bm{x}) \; .
\end{equation}
The coefficient $\frac{1}{2} \left( 1 + \frac{\delta _p}{\tau _j} \right)^{-1}$ is a first-order Taylor expansion in the parameter $\delta _p / \tau _j$ of the coefficient ${1\over2} e^{- \delta _p / \tau_j}$ in (\ref{eq:Sdef}) in the sense that $\left( 1 + \frac{\delta _p}{\tau _j} \right)^{-1}$ is obtained when one expands the denominator of $1/e^{\delta_p / \tau _j}$ to first-order in $\delta _p / \tau _j$ .  Thus, while the two limiting equations are close when all the ratios $\delta _p / \tau _j$ are small, the limiting equation derived here is overall a better approximation of the delayed system.

\section{Conclusion}\label{sec:Conclusion}

We have proven a result concerning convergence of the solution of a general SDDE driven by state-dependent colored noise to the solution of an SDE driven by white noise. The main theorem 
(Theorem \ref{theorem}) was proven using direct analysis of the SDDE without Taylor expansion and thus it improves upon a result from previous works. The 
resulting limiting equation \eqref{limiting equation} was compared to the previous results in  \cite{mcdaniel2016,pesce2013}. The noise-induced drift derived there was seen to be a first-order expansion in the ratios $\delta_p/\tau_j$ of the one found here.  The limiting equation derived here can be used as an approximation to study the dynamics of real systems modeled by the SDDE.

\section*{Acknowledgements}
A.M. and J.W. were partially supported by the NSF grants DMS 1009508 and DMS 0623941.

\bibliographystyle{plain} 
\bibliography{sde_bib2}

\begin{thebibliography}{10}

\bibitem{daprato1988}
G.~Da~Prato, S.~Kwapie{\v{n}}, and J.~Zabczyk.
\newblock Regularity of solutions of linear stochastic equations in {H}ilbert
  spaces.
\newblock {\em Stochastics}, 23(1):1--23, 1988.

\bibitem{daprato}
G.~Da~Prato and J.~Zabczyk.
\newblock {\em Stochastic Equations in Infinite Dimensions}, volume~44 of {\em
  Encyclopedia Math. Appl.}
\newblock Cambridge University Press, Cambridge, UK, 1992.

\bibitem{freidlin2004}
M.~Freidlin.
\newblock Some remarks on the {S}moluchowski-{K}ramers approximation.
\newblock {\em J. Stat. Phys.}, 117:617--634, 2004.

\bibitem{gardiner}
C.~W. Gardiner.
\newblock {\em Handbook of Stochastic Methods for Physics, Chemistry and the
  Natural Sciences}, volume~13 of {\em Springer Series in Synergetics}.
\newblock Springer-Verlag, Berlin, third edition, 2004.

\bibitem{hottovy2014}
S.~Hottovy, A.~McDaniel, G.~Volpe, and J.~Wehr.
\newblock The {S}moluchowski-{K}ramers limit of stochastic differential
  equations with arbitrary state-dependent friction.
\newblock {\em Commun. Math. Phys.}, 336(3):1259--1283, 2015.

\bibitem{hottovyEPL2012}
S.~Hottovy, G.~Volpe, and J.~Wehr.
\newblock Thermophoresis of {Brownian} particles driven by coloured noise.
\newblock {\em EPL (Europhys. Lett.)}, 99:60002, 2012.

\bibitem{janson}
S.~Janson.
\newblock {\em Gaussian {H}ilbert Spaces}, volume 129 of {\em Cambridge Tracts
  in Mathematics}.
\newblock Cambridge University Press, Cambridge, 1997.

\bibitem{kupferman2004}
R.~Kupferman, G.~A. Pavliotis, and A.~M. Stuart.
\newblock It\^o versus {S}tratonovich white-noise limits for systems with
  inertia and colored multiplicative noise.
\newblock {\em Phys. Rev. E}, 70:036120, 2004.

\bibitem{kurtz91}
T.~G. Kurtz and P.~Protter.
\newblock Weak limit theorems for stochastic integrals and stochastic
  differential equations.
\newblock {\em Ann. Probab.}, 19(3):1035--1070, 1991.

\bibitem{mcdaniel2016}
A.~McDaniel, O.~Duman, G.~Volpe, and J.~Wehr.
\newblock An {SDE} approximation for stochastic differential delay equations
  with state-dependent colored noise.
\newblock {\em Markov Process. Relat.}, 22(3):595--628, 2016.

\bibitem{Mijalkov}
M.~Mijalkov, A.~McDaniel, J.~Wehr, and G.~Volpe.
\newblock Engineering sensorial delay to control phototaxis and emergent
  collective behaviors.
\newblock {\em Phys. Rev. X}, 6(1):011008, 2016.

\bibitem{nelson}
E.~Nelson.
\newblock {\em Dynamical Theories of {B}rownian Motion}.
\newblock Princeton University Press, Princeton, N.J., 1967.

\bibitem{oksendal}
B.~{\O}ksendal.
\newblock {\em Stochastic Differential Equations}.
\newblock Universitext. Springer-Verlag, Berlin, {S}ixth edition, 2003.
\newblock {An} Introduction with Applications.

\bibitem{pavliotis2005}
G.~A. Pavliotis and A.~M. Stuart.
\newblock Analysis of white noise limits for stochastic systems with two fast
  relaxation times.
\newblock {\em Multiscale Model. Simul.}, 4(1):1--35, 2005.

\bibitem{pavliotis}
G.~A. Pavliotis and A.~M. Stuart.
\newblock {\em Multiscale Methods: Averaging and Homogenization}, volume~53 of
  {\em Texts in Applied Mathematics}.
\newblock Springer, New York, 2008.

\bibitem{peligrad2005}
M.~Peligrad and S.~Utev.
\newblock A new maximal inequality and invariance principle for stationary
  sequences.
\newblock {\em Ann. Probab.}, 33(2):798--815, 2005.

\bibitem{pesce2013}
G.~Pesce, A.~McDaniel, S.~Hottovy, J.~Wehr, and G.~Volpe.
\newblock Stratonovich-to-{I}t{\^o} transition in noisy systems with
  multiplicative feedback.
\newblock {\em Nat. Commun.}, 4:2733, 2013.

\bibitem{protter}
P.~E. Protter.
\newblock {\em Stochastic Integration and Differential Equations}, volume~21 of
  {\em Stochastic Modelling and Applied Probability}.
\newblock Springer-Verlag, Berlin, second edition, 2005.

\bibitem{revuz}
D.~Revuz and M.~Yor.
\newblock {\em Continuous Martingales and {B}rownian Motion}, volume 293 of
  {\em Grundlehren der Mathematischen Wissenschaften [Fundamental Principles of
  Mathematical Sciences]}.
\newblock Springer-Verlag, Berlin, third edition, 1999.

\bibitem{Volpe2016}
G.~Volpe and J.~Wehr.
\newblock Effective drifts in dynamical systems with multiplicative noise: a
  review of recent progress.
\newblock {\em Rep. Prog. Phys.}, 79(5):053901, 2016.

\end{thebibliography}

\end{document}